\documentclass[11pt]{amsart}
\usepackage{amsthm, amsmath, amssymb, amsfonts, pmat}
\usepackage[pdftex]{hyperref}
\hypersetup{pdfstartview=FitH,	pdfpagemode=UseNone}

\newcommand{\F}{\mathbb F}
\renewcommand{\H}{\mathbb H}
\renewcommand{\P}{\mathbb P}
\newcommand{\N}{\mathbb N}
\newcommand{\Part}{\operatorname{Part}}
\newcommand{\rk}{\operatorname{rk}}
\newcommand{\dist}{\operatorname{d}}
\newcommand{\coeff}{\operatorname{coef}}
\newcommand{\Supp}{\operatorname{S}}
\newcommand{\supp}{\operatorname{s}}
\newcommand{\poly}{\operatorname{poly}}
\newcommand{\monomials}{\operatorname{monom}}
\newcommand{\bS}{\operatorname{bS}} %block support
\newcommand{\bs}{\operatorname{bs}}
\newcommand{\descend}{\operatorname{descend}}
\renewcommand{\sp}{\mathsf{s}} %sparsity
\newcommand{\Sp}{\mathsf{S}}
\newcommand{\Span}{\operatorname{span}}
\newcommand{\paths}{\operatorname{paths}}

\newcommand{\p}{polynomial }
\newcommand{\ml}{multilinear }
\newcommand{\n}{[n]}
\renewcommand{\d}{\delta}
\renewcommand{\o}[1]{\overline{#1}}
\newcommand{\lis}[3]{{#1}_1 #2 {#1}_2 #2\dots #2 {#1}_{#3}}
\newcommand{\Dx}{D(\overline{x})}
\newcommand{\lc}{$\ell$-concentration }

\newcommand{\abs}[1]{\lvert #1 \rvert}
\newcommand{\norm}[1]{\lVert #1 \rVert}
\newcommand{\ceil}[1]{\lceil #1 \rceil}
\newcommand{\phim}{\phi^m_{\delta \ell_0}}

\newcommand{\xe}{{\o{x}}^e}
\newcommand{\xei}{{\o{x}_i}^{e_i}}

\newcommand{\Null}{\operatorname{Null}}

\newtheorem{theorem}{Theorem}
\newtheorem{lemma}{Lemma}
\newtheorem{definition}[lemma]{Definition}
\newtheorem{claim}[lemma]{Claim}
\newtheorem{corollary}[lemma]{Corollary}
\newtheorem{observation}[lemma]{Observation}
\newtheorem{remark}[lemma]{Remark}
\makeatletter
\newtheorem*{rep@theorem}{\rep@title}
\newcommand{\newreptheorem}[2]{
    \newenvironment{rep#1}[1]{
         \def\rep@title{#2 \ref{##1} (restated)}
         \begin{rep@theorem}}
{\end{rep@theorem}}}
\makeatother
\newreptheorem{theorem}{Theorem}
\newreptheorem{lemma}{Lemma}
\newcommand{\C}{\mathcal{C}}
\newcommand{\Hit}{\mathcal{H}}

\newcommand{\variables}{\operatorname{var}}
\newcommand{\pivot}{\operatorname{pivot}}

\newcommand{\Dpx}{D'(\overline{x})}
\newcommand{\Epx}{E'(\overline{x})}
\newcommand{\Ex}{E(\overline{x})}
\newcommand{\Fpx}{F'(\overline{x})}
\newcommand{\Fx}{F(\overline{x})}

\newcommand{\DpRx}{D'_{R}(\o{x})}
\newcommand{\EpRx}{E'_{R}(\o{x})}
\newcommand{\DpR}{D'_{R}}
\newcommand{\EpR}{E'_{R}}
\newcommand{\ER}{E_{R}}

\newcommand{\lzc}{$\ell_0$-concentration }
\newcommand{\lzcd}{$\ell_0$-concentrated }

\newcommand{\phase}{\operatorname{phase}}

\newcommand{\phaseh}{\text{phase-}}
\newcommand{\phasesh}{\text{phase-}}

\newcommand{\tower}{\operatorname{tower}}
\newcommand{\const}{\operatorname{const}}
\newcommand{\color}{\operatorname{color}}
\newcommand{\nbd}{\operatorname{nbd}}

\newcommand{\mnl}{M_{n, \ell}}
\renewcommand{\S}{\mathcal S}
\newcommand{\T}{\mathcal T}
\newcommand{\X}{\mathcal X}
\newcommand{\mnlst}[2]{\mnl\left({\S}_#1, {\T}_#2\right)}

\newcommand{\TM}{\mathcal{M}} %transfer matrix
\newcommand{\nul}{\mathcal{N}} % null vectors

\newcommand{\Mbreak}[2]
{
\begin{pmat}({|})
#1 & #2 \cr
\end{pmat}
} 

\begin{document}

\title[Low-distance multilinear depth-$3$]{Hitting-sets for low-distance multilinear depth-$3$}

\author[Agrawal]{Manindra Agrawal}
\address{Indian Institute of Technology, Kanpur, India}
\email{manindra@cse.iitk.ac.in}

\author[Gurjar]{Rohit Gurjar}
\address{IIT Kanpur}
\email{rgurjar@iitk.ac.in}

\author[Korwar]{Arpita Korwar}
\address{IIT Kanpur}
\email{arpk@cse.iitk.ac.in}

\author[Saxena]{Nitin Saxena}
\address{IIT Kanpur}
\email{nitin@cse.iitk.ac.in}

\begin{abstract}
The depth-$3$ model has recently gained much importance, as it has become a stepping-stone to understanding general arithmetic circuits. Its restriction to {\em multilinearity} has known exponential lower bounds but no nontrivial blackbox identity tests. In this paper we take a step towards designing such hitting-sets. We define a notion of {\em distance} for multilinear depth-$3$ circuits (say, in $n$ variables and $k$ product gates) that measures how far are the partitions from a mere {\em refinement}. The $1$-distance strictly subsumes the set-multilinear model, while $n$-distance captures general multilinear depth-$3$. We design a hitting-set in time poly($n^{\delta\log k}$) for $\delta$-distance. Further, we give an extension of our result to models where the distance is large (close to $n$) but it is small when restricted to certain variables. This implies the first subexponential whitebox PIT for the sum of constantly many set-multilinear depth-$3$ circuits.

We also explore a new model of read-once algebraic branching programs (ROABP) where the factor-matrices are {\em invertible} (called invertible-factor ROABP). We design a hitting-set in time poly($\text{size}^{w^2}$) for width-$w$ invertible-factor ROABP. Further, we could do {\em without} the invertibility restriction when $w=2$. Previously, the best result for width-$2$ ROABP was quasi-polynomial time (Forbes-Saptharishi-Shpilka, arXiv 2013).

The common thread in all these results is the phenomenon of low-support `rank concentration'. We exploit the structure of these models to prove rank-concentration after a `small shift' in the variables. Our proof techniques are stronger than the results of Agrawal-Saha-Saxena (STOC 2013) and Forbes-Saptharishi-Shpilka (arXiv 2013); giving us quasi-polynomial-time hitting-sets for models where no subexponential {\em whitebox} algorithms were known before.  
\end{abstract}

\maketitle
\tableofcontents
%\vspace{-0.85cm}

\section{Introduction}
%{\em Polynomial Identity Testing} is one of the most interesting 
%problems in complexity theory. 
The problem of {\em Polynomial Identity Testing} is that of
deciding if a given polynomial is nonzero. The complexity of 
the question depends crucially on the way the polynomial is input to the PIT test.
For example, if the polynomial is given as a set of coefficients
of the monomials, then we can easily check whether the polynomial is nonzero
in polynomial time. 
The problem has been studied for different input models.
Most prominent among them is the model of arithmetic circuits. 
Arithmetic circuits are arithmetic analog of boolean circuits
and are defined over a field $\F$.
They are directed acyclic graphs, where every node
is a `$+$' or `$\times$' gate and each input gate is 
a constant from the field $\F$ or a variable from
$\o{x} = \{\lis{x}{,}{n}\}$. 
Every edge has a weight from the underlying field $\F$.
The computation is done in the natural way:
Starting with the input nodes and proceeding towards the output node.
At each step, the weight of the edge $(u, v)$ is multiplied with the output of the previous gate $u$ and then input to the next gate, $v$.
Clearly, the output gate computes a polynomial
in $\F[\o{x}]$. 
We can restate the PIT problem as: Given an arithmetic circuit $\C$, 
decide if the polynomial computed by $\C$ is nonzero in time polynomial
in the circuit size. 
Note that, given a circuit, computing the polynomial explicitly
is not possible, as it can have exponentially many monomials. 
However, given the circuit, 
it is easy to compute an evaluation of the polynomial
by substituting the variables with constants. 

Though there is no known {\em deterministic} algorithm for PIT,
there are easy randomized algorithms, e.g.\ \cite{Sch80}.
%The open question is to find a {\em deterministic} algorithm for PIT. 
These randomized algorithms are based on the theorem:
A nonzero polynomial,
evaluated at a random point, gives a nonzero value with a good probability. 
Observe that such an algorithm does not need to see 
the structure of the circuit, it just uses the evaluations;
it is a {\em blackbox} algorithm.
The other kind of algorithms, where the structure of the
input is used, are called {\em whitebox} algorithms. 
Whitebox algorithms for PIT have many known applications. E.g.\ graph matching reduces to PIT.
On the other hand, blackbox algorithms (or \emph{hitting-sets}) have connections to circuit lower bound proofs. Arguably, this is currently the only concrete approach towards lower bounds, see \cite{mul12}. See the surveys by Saxena \cite{Sax09} and Shpilka \& Yehudayoff \cite{SY10} for more
%varied references.
applications.

The PIT problem has been studied for various restricted classes of circuits. 
One such class is depth-$3$ circuits. A depth-$3$ circuit is usually
defined as a $\Sigma \Pi \Sigma$ circuit: the circuit gates are in 
three layers, the top layer has an output gate which is $+$, second layer has
all $\times$ gates and last layer has all $+$ gates. 
In other words, the polynomial computed by a $\Sigma\Pi\Sigma$ circuit is of the form 
$C(\overline{x}) = \sum_{i=1}^k a_i \prod_{j=1}^{n_i}  \ell_{ij}$, 
where $n_i$ is the number of input lines to the $i$-{th} product gate 
and $\ell_{ij}$ is a linear polynomial of the form $b_0 + \sum_{r=1}^n b_r x_r$.
An efficient solution for depth-$3$ PIT is still not known.
Recently, it was shown by Gupta et al.\ \cite{GKKS13}, that depth-3 circuits are almost as powerful as general circuits.
%finding hitting-sets or lower bounds for 
%depth-$3$ circuits is enough.
A polynomial time hitting-set for a depth-$3$ circuit implies a quasi-poly-time hitting-set for general circuits.
% and that an exponential lower bound for depth-$3$ circuits will imply a super-\p lower bound for general circuits. Hence, it will be difficult to find a PIT for depth-$3$ circuits.

On the other hand, there are exponential lower bounds for depth-$3$ {\em multilinear} circuits \cite{RY09}. Since there is a connection between lower bounds and PIT \cite{Agr05}, we can hope that solving PIT for depth-$3$ \ml circuits should also be feasible. This should also lead to new tools for general depth-$3$.

A polynomial is said to be multilinear if 
the degree of every variable in every term is at most $1$.
The circuit $C(\o{x})$ is a multilinear circuit 
if the polynomial computed at every gate is multilinear.
A polynomial time algorithm is known only for a sub-class of multilinear depth-$3$ circuits, called {\em depth-$3$ set-multilinear circuits}.
This algorithm is due to Raz and Shpilka \cite{RS05} and is whitebox.
In a depth-$3$ multilinear circuit, 
since every product gate computes a multilinear polynomial, 
a variable occurs in at most one of the $n_i$ 
linear polynomials input to it. 
Thus, each product gate naturally induces a 
{\em partition} of the variables, where
each {\em color} (i.e.~part) of the partition contains the variables
present in a linear polynomial $\ell_{ij}$.  
Further, if the partitions induced by all the 
$k$ product gates are the same then the circuit is
called a depth-$3$ set-multilinear circuit.

Agrawal et al.\ \cite{ASS13} gave a blackbox 
algorithm for this class
which works in quasi-polynomial time. 
Their approach is to view the product 
$\Dx := (\prod_{j=1}^{n_i}  \ell_{ij})_{i=1}^k$
as a polynomial over the Hadamard algebra of dimension 
$k$, $\H_k(\F)$, 
and to achieve a \emph{low-support concentration} in it. 
Low-support concentration means that all the coefficient vectors
in $\Dx$
are linearly dependent on low-support coefficient vectors. 
We define a new class of circuits called
{\em multilinear depth-$3$ circuits
with $\delta$-distance}. We show low-support concentration for this model. 
To achieve that we use some
deeper combinatorial properties of our model.
We also use an improved version of a combinatorial property 
of the transfer matrix 
%, which is an improve version of a similar result
from \cite[Theorem 13]{ASS13} and we present it with a simplified proof.
Recently, Forbes et al.\ \cite{FSS13} have also improved 
\cite{ASS13}. But their methods apply only to set-\ml formulas, and not even to $1$-distance circuits.

A multilinear depth-$3$ circuit has $\delta$-distance
if there is an ordering 
on the partitions induced by the product gates, 
say $(\lis{\P}{,}{k})$, such that for any color 
in the partition $\P_i$, there exists a set of $\delta-1$ 
other colors in $\P_i$ such that the set of variables in 
the union of these $\delta$ colors are {\em exactly} partitioned
in the upper partitions, i.e.\ $\{\lis{\P}{,}{i-1}\}$.
Intuitively, the distance measures how far away are the partitions from a mere \emph{refinement} sequence of partitions, $\lis{\P}{\le}{k}$.
Our first main result gives a blackbox test for this class of circuits (Section~\ref{sec:dDist}).
%\vspace{-0.1cm}
\begin{theorem}\label{thm:dDistHS}
Let $C(\o{x})$ be a $\d$-distance depth-$3$, $n$-variate \ml circuit with top fan-in $k$. Then
there is a $n^{O(\d \log k)}$-time hitting-set for $C(\o{x})$. 
\end{theorem}
%\vspace{-0.1cm}
Note that the running time becomes quasi-polynomial 
 when $\delta$ is poly-logarithmic in $nk$.
Till now, no subexponential time test was known for this class, 
even in the whitebox setting.
Also, observe that the set-multilinear class is strictly subsumed 
in the class of $1$-distance circuits. E.g.\ a circuit, whose product gates
induce two different partitions
$\P_1 = \{\{1\}, \{2\}, \dots, \{n\}\}$ and 
$\P_2 = \{ \{1,2\} , \{3,4\} , \dots, \{n-1,n\} \}$, 
has $1$-distance but is not set-multilinear. 
So, poly-logarithmic $\delta$ is a significant improvement 
from set-multilinear. 
On the other hand, general multilinear depth-$3$ circuits can have at most
$n$-distance.
So, our result is also a first step towards \ml depth-$3$ circuits.

Our second result further generalizes this class
to {\em multilinear depth-$3$ circuits having $m$ base sets
with $\delta$-distance}. A circuit is in this class
if we can partition the set of variables
into $m$ base sets, such that when restricted to any of
these base sets, the circuit has $\delta$-distance.
E.g.\ consider a circuit $C$, whose product gates induce two partitions
$\P_1 = \{ \{1,2\} , \{3,4\} , \dots, \{n-1,n\}  \}$
and $\P_2 = \{ \{ 2,3 \} , \{ 4,5 \} , \dots, \{n, 1\}  \}$. Clearly, $C$ is $(n/2)$-distance. But, when restricted to any of the two base sets
 $B_1 = \{ 1,3, \dots, n-1 \}$ and 
$B_2 = \{ 2,4, \dots, n\}$, $C|_{B_i}$ has $1$-distance 
(in fact, it is set-multilinear).
We give a quasi-polynomial time blackbox test
for this class, when $m$ and $\delta$ are poly-logarithmic 
and the base sets are {\em known} (Section~\ref{sec:baseSets}).
%\vspace{-0.1cm}
\begin{theorem}
\label{thm:baseSetsHS}
If $C(\o{x})$ is a depth-$3$ multilinear circuit, with top fan-in $k$, having $m$ base sets (known) with $\delta$-distance, then there is a $n^{O(m \delta \log k)}$-time hitting-set for $C$. 
\end{theorem}
%\vspace{-0.1cm}
These results generalize to suitable higher-depth circuit models. 
But we focus, in this paper, only on the depth-$3$ case to avoid unnecessary
complicated notations. 
Theorem~\ref{thm:baseSetsHS} also implies a whitebox PIT, for the
sum of $c$ (constant) set-multilinear depth-$3$ circuits (top fan-in $k$), with time complexity
$n^{O(n^{1-\epsilon} \log k)}$, where $\epsilon := 1/2^{c-1}$ % and $k$ is the top fan-in
(see Appendix~\ref{sec:sumSetMult}).

Our third result expands the realm of low-support concentration to
\ml variants of Arithmetic Branching Programs (ABP).
An ABP is another interesting model of computing polynomials. 
It consists of a ditected acyclic graph with a source and a sink. 
The edges of the graph have polynomials as their weights.
The weight of a path is the product of the weights of the edges
present in the path. 
The polynomial computed by the ABP
is the sum of the weights of all the paths from the source to the sink.
It is well known that for an ABP, the underlying graph 
can seen as a layered graph such that all paths from the source to
the sink have exactly one edge in each layer. 
And the polynomial computed by the ABP can be written as 
a \emph{matrix product}, where each matrix corresponds to a layer. 
The entries in the matrices are weights of the corresponding edges. 
The maximum number of vertices in a layer, 
or equivalently, the dimension of the corresponding matrices 
is called the \emph{width }of the ABP.
Ben-Or \& Cleve \cite{BOC92} have shown that 
a polynomial computed by a formula of logarithmic depth and constant fan-in,
can also be computed by a width-$3$ ABP.
Moreover, Saha et al.\ \cite{SSS09} showed that PIT for depth-$3$ 
circuits reduces to PIT for width-$2$ ABP.
Hence, constant width ABP is already a strong model. 
Our results are for constant width ABP with some natural restrictions. 
 
An ABP is a read once ABP (ROABP) if the entries in the 
different matrices come from disjoint sets of variables. 
Forbes et al.\ \cite{FSS13} recently gave a quasi-polynomial time
blackbox test for ROABP, when 
the entries in each matrix are essentially {\em constant}-degree univariate polynomials. 
Their approach, too, involves low-support concentration. 

Our result is for ROABP with a further restriction. 
We assume that all the matrices in the matrix product, except
the left-most and the right-most matrices, are invertible. 
We give a blackbox test for this class of ROABP. 
In contrast to \cite{FSS13}, our test works in \emph{polynomial time}
if the dimension of the matrices is constant;
moreover, we can handle univariate factor matrices with any degree.
%\vspace{-0.2cm}
\begin{theorem}[Informal version]
\label{thm:ROABPHS}
Let $C(\o{x}) = D_0^T (\prod_{i=1}^{d} D_i ) D_{d+1}$
be a polynomial such that $D_0 \in \F^w[{x}_{j_0}]$ and $D_{d+1} \in \F^w[{x}_{j_{d+1}}]$
and for all $i \in [d]$, $D_i \in \F^{w \times w}[x_{j_i}]$ is an invertible matrix (order of the variables is unknown).
Let the degree bound on $D_i$ be $\delta$ for $0 \leq i \leq d+1$.
Then there is a $\poly((\delta n)^{w^2})$-time hitting-set for $C(\o{x})$. 
\end{theorem}
%\vspace{-0.1cm}
The proof technique here is very different from the first two theorems (now, we show rank concentration over a {\em non-commutative} algebra).
Our algorithm works even when the factor matrices have their entries as general sparse 
polynomials (still over disjoint sets of variables) instead of 
univariate polynomials (see detailed version in Section~\ref{sec:ROABP}).
%Theorem~\ref{thm:ROABPHS-b}). 
Running time in this case is quasi-polynomial.
We remark that the invertibility seems to restrict the computing power of ABP.
If the matrices are $2 \times 2$, we do not need the assumption
of invertibility (see Theorem~\ref{thm:ROABP22HS}, Appendix~\ref{app:2ROABP}). 
So, for width-$2$ ROABP our results are strictly stronger than \cite{FSS13}.
 
%\vspace{-0.25cm}
\subsection{Main idea of Theorem~\ref{thm:dDistHS}}
As mentioned earlier, the basic idea is to show
low-support concentration in $D(\o{x})$, but by an 
efficient shift (Lemma~\ref{lem:dDistlConc}).
While showing low-support concentration in
set-multilinear case, the key idea of \cite{ASS13}
was to identify low degree subcircuits of $D(\o{x})$
which have `true coefficients', i.e.\
each of these subcircuits is such that, 
when multiplied by an appropriate constant vector, its coefficients 
become the coefficients of $\Dx$. 
%only a constant vector needs to be multiplied with all of the coefficients in a subcircuit
%to obtain the coefficients of the original circuit.
%their coefficients are different from 
%the original circuit only by a multiplicative
%factor.
Then they show that
an efficient shift can ensure low-support concentration
in these subcircuits. 
The final step is to argue that low-support concentration
in the subcircuits translates to low-support concentration in 
the actual circuit. 

Such {\em true} subcircuits,
in a general multilinear circuit,
may have a high degree. % (hence, a large support).
In the case of small distance circuits, there exist true subcircuits
with low degree in the last partition.
But they may have many high degree monomials in the upper (other)
partitions.
That is not good, because the degree of the subcircuit affects the hitting-set size.
%To prove low-support concentration, all of these high-support coefficients from upper partitions need to be shown to be low-support concentrated. 
%And all of these high-support coefficients are from the upper partitions.
%We would use the observation that the high degree coefficients are only from the upper partitions. 
%Then, the monomials with low-support in all 

This inspired us to prove concentration in phases, i.e.\ by induction on $k$.
We divide the coefficients 
into $k$ phases (one corresponding to each partition).
Phase-$j$ coefficients are those which have nonzero values only in 
the coordinates $\{1, 2, \dots, j\}$.
We show concentration in successive phases (Lemma~\ref{lem:j-1Impliesj}). 
In each phase, the concentration in the previous phases is assumed. 

How do we isolate phase-$j$ coefficients?
We take an appropriate partial derivative of the circuits,
which ensures that the values in the coordinates $\{j+1, j+2, \dots, k\}$ are 
zero. 
Since the partial derivatives add to the complexity, they should be of small order
(Observation~\ref{obs:pivot}).

Now, we identify subcircuits of this \emph{derivative polynomial},
which have low degree in the $j$-th coordinate.
The high degree in other coordinates does not matter as we assume that there is already concentration in the previous phases. 
We show low-support concentration in these subcircuits (Lemma~\ref{lem:subcircuitl0}; it is the most technical part of the proof). 
This
in turn implies low-support concentration in the derivative polynomial 
(Lemma~\ref{lem:subcircuit}) by another induction on the {\em neighborhoods} (defined in Section \ref{subsec:deltaDistance}). 
To show low-support concentration in a subcircuit
we need to show some combinatorial properties of the {\em transfer matrix} (Lemma~\ref{lem:rowComb}).
Finally, the concentration in the derivative 
polynomial implies concentration among the phase-$j$ coefficients
of $D(\o{x})$ (Lemma~\ref{ConcInMonomialsIncludingR}).

%Thus, 
%the concentration we get is of a bit higher support. 
Carrying the argument from phase-$1$ to phase-$k$ finishes the proof. The cost of this is only quasipoly in $k$ (because the dimension of the underlying vector space is $k$ and there is an implicit `doubling effect' in Lemma \ref{lem:rowComb}), but {\em exponential} in $\delta$ (because the true subcircuits have degree $\delta$, which makes the Kronecker map expensive). Any improvement in the latter would lead to the first nontrivial PIT for multilinear depth-$3$ circuits.

%We should also point out that in the
%intemediate step, we prove concentration in a general low degree
%polynomial instead of a set-multilinear low degree polynomial, as done by \cite{ASS13}. 
%This is another place where we generalized the ideas of \cite{ASS13}. 

%\vspace{-0.27cm}
\section{Preliminaries}
%%\vspace{-0.27cm}
$\n$ denotes the indices $1$ to $n$. Let $2^{[n]}$ denote the set of all subsets of $[n]$.
$\Part(S)$ denotes the set of all possible partitions of the set $S$. Elements in a partition are called {\em colors}. The \emph{support} of a monomial is the set of variables that have degree $ \ge 1 $ in that monomial. The \emph{support size }of the monomial is the cardinality of its support. $\F^{m \times n}$ represents the set of all $m \times n$ matrices over the field $\F$.
$\F^{S \times T}$, where $S$ and $T$ are sets, represents the set of all $\abs{S} \times \abs{T}$ matrices over the field $\F$, indexed by the elements of $S$ and $T$.
The matrices in this paper are often indexed by subsets of $\n$.

%The non-trivial depth-$3$ circuits are $\Sigma\Pi\Sigma$ circuits. They compute the summation 
%of product of linear terms. A polynomial computed by a $\Sigma\Pi\Sigma$ circuit is of the form 
%$C(\overline{x}) = \sum_{i=1}^k a_i \prod_{j=1}^{n_i}  \ell_{ij}$, where $n_i$ is the number of input lines to $i$-th product gate and $\ell_{ij}$ is a linear polynomial of the form $b_0 + \sum_{i=1}^n b_i x_i$. 
%Since we will study only polynomial sized circuits, $k$ and $n_i$ are both polynomial in $n$. 
%A polynomial is said to be \emph{multilinear} if the degree of every variable in every term is at most $1$. 
A multilinear \p can be represented as $\sum_{S\subseteq [n]} a_Sx_S$, where $x_S$ is the monomial $\prod_{i\in S}x_i$. We will sometimes use the notation $ S $ for the multilinear monomial $ x_S $. 
%The circuit $C(\overline{x})$ is a \emph{multilinear circuit} if the polynomial computed at every gate is multilinear. Hence, in  a depth-$3$ multilinear circuit, since every product gate computes a multilinear polynomial, a variable occurs in at most one of the $n_i$ linear polynomials input to it. Thus, each product gate naturally induces a {\em partition} of the variables. Each \emph{color} of the partition corresponds to the variables in a single linear \p. E.g.\ Let $\P = \lis{X}{\sqcup}{n_i}$ be a partition on $[n]$. It partitions the $n$ variables into $n_i$ colors: $\lis{X}{,}{{n_i}}$. Each product gate in a multilinear circuit introduces one such partition.

%Let $\P$ be a partition on $\{x_1, x_2, \dots, x_n\}$. 
%A multilinear polynomial $Q(\overline{x})$ is \emph{$\P$-set-multilinear }if there is at most one variable from each color $X_j$ in the support of each monomial of the \p $Q$. E.g.\ if $\P = \{x_1, x_2\} \sqcup \{x_3, x_4\} \sqcup \{x_5, x_6\}$, then $ x_1x_4x_5$ and $x_1x_4$ are $\P$-set-multilinear; $x_1x_2$ is not. A circuit is $\P$-set-multilinear if the \p computed at each of its gates is $\P$-set-multilinear. A depth-$3$ \ml circuit is $\P$-set-multilinear if every product gate respects the same partition $\P$. Observe that every %$\P$-set-multilinear \p is multilinear. Every 
%multilinear polynomial is a $\P$-set-multilinear \p for the partition $\P = \{x_1\}\sqcup\{x_2\}\sqcup \dots \sqcup \{x_n\}$.

Henceforth, we will only discuss polynomials computed by depth-$3$ \ml circuits, unless explicitly stated otherwise.
The \p computed by a depth-$3$ circuit, $C(\overline{x}) = \sum_{i=1}^k a_i \prod_{j=1}^{n_i}  \ell_{ij}$ can also be written as the inner product of the vector ${\o{a}} = (\lis{a}{,}{k})^T $ and $D(\overline{x})$, where $D(\overline{x})$ is a polynomial over the $k$-dimensional Hadamard algebra $\H_k(\F)$.
The Hadamard algebra $\H_k(\F)$ is defined as $(\F^k, +, \star)$, where 
$+$ and $\star$ are coordinate-wise addition and multiplication.
%$D(\overline{x})\in \H_k(\F)[\lis{x}{,}{n}]$. 
The $i$-th coordinate of $D(\overline{x})$ is $\prod_{j=1}^{n_i}  \ell_{ij}$. Hence, $C(\overline{x}) = \overline{a}^T D(\overline{x}) $. Let $ \coeff_D(x_S) $ denote the coefficient of the monomial $ x_S $ in the \p $ \Dx $. The coefficients $\{\coeff_D(x_S) \mid S \subseteq \n\}$ of $D(\o{x})$ form a $(\le  k) $-dimensional vector space over the base field $ \F $. Our rough plan is to show that this vector space is spanned by the coefficients of the `low-degree' monomials. We thus define \emph{$\ell$-concentration} for a polynomial $ \Dx $ whose coefficients are vectors from a $ k $-dimensional vector space.

\begin{definition}[$\ell$-concentration]
Polynomial $ \Dx \in \H_k(\F)[\lis{x}{,}{n}]$ is $ \ell $-concentrated if
$
\rk_\F \{\coeff_D (x_S) \mid S \subseteq [n], \abs{S} < \ell\} = \rk_\F \{\coeff_D (x_S) \mid S \subseteq [n]\}.
$
\end{definition}

The following Lemma from \cite{ASS13} says that a \p $C(\o{x})$, with an $\ell$-concentrated polynomial $\Dx$, has a hitting set
%gives a hitting-set for the polynomial $C(\o{x})$ 
(for a proof, see Section~\ref{app:dDist}).
%\vspace{-0.1cm}
\begin{lemma}\label{lem:hsFromlConc}%1
If $D(\overline{x})$ is $\ell$-concentrated, then there is a $n^{O(\ell)}$-time hitting-set for $C(\overline{x})$.
\end{lemma}
%\vspace{-0.1cm}
However, observe that low-support concentration does not exist in all polynomials $\Dx$. E.g.\ in the polynomial $ \Dx = \overline{c}\cdot x_1x_2\dots x_n$, there are no low-support monomials. To counter this problem, the polynomial is \emph{shifted}. Each input $ x_i $ to the polynomial is replaced with $ x_i + t_i $, where $ t_i $s are symbolic constants adjoined to the base field $ \F $. Now, the input field is considered to be the field of fractions $ \F(\overline{t}) $, where $\o{t}= \{ t_1, t_2, \dots, t_n \}$. Since after shifting, the coefficients of high-support monomials contribute an additive term to the coefficients of the low-support monomials, we can hope to prove \lc over this field of fractions. We will use the notation $ D'(\overline{x}) $ as well as $ D(\o{x} + \o{t}) $ to mean $ D(x_1 + t_1, x_2 + t_2, \dots, x_n + t_n)$. For the example above, $ D'(\o{x}) = \sum_{S \subseteq [n]} \overline{c} \cdot t_{\bar{S}} x_S$. The dependence of $\coeff_D(x_S)$
% the coefficient of $ x_S $ 
 over the field $ \F(\overline{t}) $ is given by $\coeff_{D'}(x_S) =  \overline{c}\cdot t_{\bar{S}} = t_{\bar{S}}\coeff_{D'}(x_{\varnothing})$. Thus, in the above example, $D'(\overline{x})$ is $ \ell $-concentrated for $\ell = 1$.

%In this paper, $\ell = O(\log n)$. 
We conjecture that $O(\log n + \log k)$-concentration can be proven for all multilinear circuits after an appropriate shift. Agrawal, Saha \& Saxena (\cite{ASS13}) have proven $O(\log k)$-concentration of set-multilinear circuits after an efficient shift.
Here, we study low-support concentration for more general models, by developing stronger techniques.
%Here we follow their approach and prove low-support concentration for a more general class. 
%In this paper, we will prove \lc for $\ell = O(\log n)$ in some special circuits. %In this paper,
%We prove \lc for a subclass of \ml circuits

%\vspace{-0.2cm}
\subsection{General Approach}\label{subsec:generalApproach}
How do we prove that the high-support coefficients of $D'$ are dependent on the low-support coefficients? Do we even know that a given high-support coefficient is dependent on other coefficients? The reason we believe such a dependency exists is because of shifting. A monomial $x_S$, when the circuit is shifted, contributes its coefficient, denoted by $u_S$, to the coefficients of all of its subsets: $u'_T := \coeff_{D'}(x_T)= \sum_{S \supseteq T} u_S t_{S \setminus T}$. The \p~$D'(\o{x} - \o{t}) = \Dx$, i.e.\ by shifting every variable $x_i$ by $-t_i$ in $\Dpx$, we get back $\Dx$. Thus, $u_T = \sum_{S \supseteq T} u'_S t_{S \setminus T} (-1)^{\abs{S \setminus T}}$. 
This can be represented as
$
U = U' \cdot \TM 
$, 
where $U \in \F^{[k] \times 2^{[n]}}$ and $U' \in (\F(\o{t}))^{[k] \times 2^{[n]}}$ represent the coefficients in the polynomials $\Dx$ and $\Dpx$ respectively. The matrix $\TM$ has $(S, T)$-th entry
\[
\TM(S, T) =
\begin{cases}
t_{S \setminus T} (-1)^{\abs{S \setminus T}} & \text{if } T \subseteq S,\\
0 & \text{otherwise.}
\end{cases}
\]

Equivalently, we can write $\TM = A^{-1} M A$, where 
$
M(S, T) =
1 \text{ if } T \subseteq S 
\text{ and }
0\text{ otherwise,}
$
and
$A$ is a diagonal $2^{[n]} \times 2^{[n]}$ matrix where the $(T, T)$-th entry is $(-1)^{\abs{T}}t_T^{-1}$.
% and the $(S, T)$-th entry of the matrix $M$ is

To analyze the dependencies among vectors in $U'$
we take a dependency for the vectors in $U$ and \emph{lift} it. 
Suppose a dependency for the $U$ vectors is: $\sum_T \alpha_T u_T = 0, \alpha_T \in \F$.
Now, replace the coefficients $u_T$ with an equivalent expression in terms of the coefficients $u'_S$. 
We get $\sum_T \alpha_T \sum_{S \supseteq T} u'_S t_{S \setminus T} (-1)^{\abs{S \setminus T}} = 0$.
% Observe that the high-support coefficients in $\Dpx$ participate in this dependency. 
The coefficient of $u'_S$ in the dependency is $\sum_{T \subseteq S} \alpha_T t_{S \setminus T} (-1)^{\abs{S \setminus T}}$. The coefficient for each $u'_S$ is nonzero, and thus, it participates non-trivially in a dependency.

There are two problems here. First, we do not directly get dependencies which show that
high-support coefficients are in the span of low-support coefficients. 
%The first problem is that only a single high-support coefficient should participate in a single dependency, so that it is dependent on the low-support coefficients. 
There must be one such dependency for each high-support $u'_S$, which
shows that it is in the span of low-support coefficients. 
Existence of such dependencies will be shown in the later sections,
by considering all the dependencies of the $U$ vectors (null vectors of $U$) and their lifts.
%The solution to this problem will be given in subsection \ref{fullRankToIdentity}.

The other problem is that, even if low-support concentration does exist, we are substituting $(x_i + t_i)$s instead of $x_i$s and 
%then substituting $x_i$s with constants(Lemma~\ref{lem:hsFromlConc}). 
then computing low-support coefficients (Lemma~\ref{lem:hsFromlConc}).
The coefficients themselves can be exponentially large polynomials in the variables $\o{t}$, and thus cannot
be computed efficiently.
%The variables $t_i$s replace the variables $x_i$s. 
%This does not reduce the complexity of polynomial testing. 
%We use the standard approach to resolve this problem \ref{kronecker}.
It turns out that we can substitute the shift variables with an 
efficient univariate map and show low-support concentration. 

%\vspace{-0.2cm}
\subsection{Kronecker substitution}
\label{subsec:kronecker}
The shift variables $t_i$s are replaced with powers of a single variable $t$. 
Let us say, the degree of the variable $t$ is upper bounded by some function $g(n)$. Then the $t$-degree of the polynomial computed by the depth-$3$ multilinear circuit is bounded by $ng(n)$. The computation is thus efficient when the shift is univariate and `small'.
%Consider the following definition. 
%But, after applying the substitution $\phi: \{\lis{t}{,}{n}\} \rightarrow \{t^i\}_{i = 0}^{g(n)}$, the coefficient of $u_S \in \F(t)$  may become 0. The terms of the coefficient before the substitution did not cancel each other because every monomial had a different monomial in $t_i$s. If we can still ensure that the monomials map to different powers of $t$, then the terms of the coefficient cannot cancel each other. But it can easily be shown that only a single univariate substitution that separates all the monomials of an arbitrary \ml $n$-variate polynomial requires the degree of $t$ to be exponential. Hence, we use a set of small degree substitutions such that at least one of the substitutions separates the monomials of a given polynomial.
%\begin{definition}[$(f(n), g(n))$-shift set]
%A set of $f(n)$ shifts of the set of variables $\{\lis{x}{,}{n}\}$ by powers of $t$ where the highest power of $t$ is $\le g(n)$.
%\end{definition}
%\begin{definition}[$(f(n), g(n))$-monomial map]
%A map $\phi \colon \{\lis{t}{,}{n}\} \to \{t^i\}_{i=1}^{\infty}$ is called a $(f(n), g(n))$- 
%monomial map if we can generate a set of $f(n)$-many maps such that $\phi$ belongs to that set %%and the highest power of $t$ in their range is $\le g(n)$.
%\end{definition}
The univariate map we use 
will need to separate all ($\leq \ell $) support monomials for some small $\ell$, i.e.\
all ($\leq \ell $) support monomials should be mapped to distinct powers of $t$ 
(the map can be seen as acting on monomials in the natural way e.g.\ $\phi(t_1 t_2) = \phi(t_1)\phi(t_2)$). 
This map will be denoted by $\phi_{\ell}$.
%The map which separates all ($\leq \ell $) support monomials will be denoted by $\phi_\ell$. 
For the time complexity of generating such maps, see Lemma~\ref{lem:kronecker} (Appendix~\ref{app:dDist}).
%In the following lemma, we show how to get such a map. 
%In this section we are only concerned with
%multilinear monomials, but we are writing a result for general monomials 
%which will be useful in later sections. 
%Thus, if the \p is \lcd with the variables $\{\lis{t}{,}{n}\}$ and if the number of monomials involved in each dependency is $<a$ then, there exists a $(na^2, na^2\log (na^2))$-shift set that gives \lc. Each of the substitutions in the $(na^2, na^2\log (na^2))$-shift set need to be checked for \lc, after which theorem \ref{lem:hsFromlConc} can be applied.
We now describe the effect of the shifting map on the
final time complexity.
The notation $D(\o{x} + \phi(\o{t}))$ will mean $D(x_1 + \phi(t_1) , x_2 + \phi(t_2) , \dots , x_n + \phi(t_n))$. 
%\vspace{-0.2cm}
\begin{lemma}[$\ell$-Concentration to hitting-sets]
\label{lem:finalHS}
If for a polynomial $ D(\o{x}) \in \H_k(\F)[\o{x}]$, there exists a set of $f(n)$-many maps from 
$\o{t}$ to $\{t^i\}_{i=1}^{g(n)}$ such that for at least one of the maps $\phi$, the shifted polynomial
$D(\o{x}+\phi(\o{t}))$ has $\ell$-concentration, then $C(\o{x}) = \o{a}^T D(\o{x})$, for any $\o{a} \in \F^k$, has an $n^{O(\ell)} f(n) g(n)$-time hitting-set. \emph{(Proof in Appendix~\ref{app:dDist}.)}
\end{lemma}

%Hence, the time complexity to check the zeroness of the \p $C(X)$ depends on $a$, the number of monomials in each dependency. If the number of monomials in each dependency were less than $n^{O(\ell_0)}$,%, i.e.\ $a = n^{O(\ell)}$,
%then the time complexity to check the zeroness of the \p $C(X)$ is $n^{O(\ell) + O(\ell_0)}$.
%In the following model of depth-$3$ \ml circuits, where the distance for the partition sequence (to be defined soon) is `small', the number of monomials involved in each dependency can indeed be bounded.

%Now, in the following section we define the class of depth-$3$ multilinear circuits
%for which we can achieve low-support concentration after an efficient shift.

Now that we are clear about the general technique, we can proceed to proving \lc in some interesting models.

%\vspace{-0.25cm}
\section{Low-distance multilinear depth-$3$ circuits: Theorem~\ref{thm:dDistHS}}
\label{sec:dDist}
The main model for which we study low-support concentration is depth-3 multilinear circuits with `small distance'.

%%\vspace{-0.25cm}
\subsection{$\delta$-distance circuits}\label{subsec:deltaDistance}

Each product gate in a depth-$3$ \ml circuit induces a partition on the variables. Let these partitions be $ \lis{\P}{,}{k} $.

\begin{definition}[Distance for a partition sequence, $ \dist(\P_1,\dots, \P_k) $]
\label{def:distance}
Let $\P_1, \P_2, \dots, \P_k \in \Part (\n)$ be the $k$ partitions of the variables $\{x_1, x_2, \dots, x_n\}$. Then $\dist(\P_1, \P_2, \dots, \P_k) =: \d$ if 
% there exists an ordering of the partitions (wlg $\lis{\P}{,}{m}$) such that 
 $ \forall i \in \{2,3,\dots , k\}, \; \forall\text{colors } Y_1 \in \P_i, \; \exists Y_2, Y_3, \dots, Y_{\d'} \in \P_i  (\d' \le \d)$ such that $\lis{Y}{\cup}{\d'}$ equals a union of some colors in $ \P_j, \forall j\in [i-1]$.
\end{definition}

In other words, in every partition $\P_i$,
each color $Y_1$ has a `friendly neighborhood' $\{\lis{Y}{,}{\d'}\}$,
consisting of at most $\d$ colors,
which is exactly partitioned in the `upper partitions'.
We call $\P_i$, an {\em upper} partition relative to $\P_j$ 
(and $\P_j$, a {\em lower} partition relative to $\P_i$),
if $i < j$.
For a color $X_a$ of a partition $\P_j$, let $\nbd_j (X_a)$ denote its friendly neighborhood.
The friendly neighborhood $\nbd_j (x_i)$ of a variable $x_i$ in a partition $\P_j$ is defined as $\nbd_j (\color_j(x_i))$, where $\color_j(x_i)$ is the color in the partition $\P_j$ that contains the variable $x_i$.
The friendly neighborhood $\nbd_j(\{x_{i}\}_{i \in {\mathcal I}})$ of a set of variables $\{x_{i}\}_{i \in {\mathcal I}}$ in a partition $\P_j$ is given by $\bigcup_{i \in {\mathcal I}} \nbd_j(x_i)$.

%The friendly neighborhood $\nbd_j(\{X_{a_i}\}_{i \in {\mathcal I}})$ of a set of colors $\{X_{a_i}\}_{i \in {\mathcal I}}$ in a partition $\P_j$ and the friendly neighborhood $\nbd_j(\{x_{i}\}_{i \in {\mathcal I'}})$ of a set of variables $\{x_{i}\}_{i \in {\mathcal I'}}$ in a partition $\P_j$ are given by $\bigcup_{i \in {\mathcal I}} \left\{\nbd_j(X_{a_i})\right\}$ and $\bigcup_{i \in {\mathcal I'}} \left\{\nbd_j(x_i)\right\}$ respectively.

% Let $ k_i $ product gates follow the partition $ \P_i $. Hence the total number of product gates is $ k:= \sum_{i = 0}^m k_i $.

\begin{definition}[$ \d $-distance circuits]
A \ml depth-$3$ circuit $ C $
% is said to have 
has $ \d $-distance if its product gates can be ordered to correspond to a partition sequence $(\P_1,\dots,\P_k)$ with $\dist (\lis{\P}{,}{k}) \le \d$.

The corresponding $\Pi\Sigma$ circuit $D(\o{x})$ over $\H_k(\F)$ is also said to have $\d$-distance.
%A multilinear \p over the Hadamard algebra of dimension $k$,
%$\Dx \in \H_k(\F)[\o{x}]$ is said to have $ \d $-distance if
%the partition sequence induced by its $k$ coordinates
%corresponds to the partition sequence  $(\lis{\P}{,}{k}) $ with $\dist (\lis{\P}{,}{k}) \le \d$.
\end{definition}

%Let there be $k_i$ product gates that correspond to the partition $\P_i$. Hence, the total number of product gates in the circuit is $\sum_{i = 1}^m k_i =: k$.

%\begin{observation}
Every depth-$3$ \ml circuit is thus an $n$-distance circuit. A circuit with a partition sequence, where the partition $\P_i$ is a refinement of the partition $\P_{i+1}, \forall i \in [k-1]$, exactly characterizes a $1$-distance circuit. All depth-$3$ \ml circuits have distance between $1$ and $n$.
%\end{observation}
Also observe that the circuits with $1$-distance subsume set-\ml circuits.

\noindent \textbf{Friendly neighborhoods -} To get a better picture, we ask: Given a color $X_a$ of a partition $\P_j$ in a circuit $\Dx$,
how do we find its friendly neighborhood $\nbd_j(X_a)$?
Consider a graph $G_j$ which has the colors of the partitions $\{ \P_1, \P_2, \dots, \P_j \}$, 
as its vertices. 
For all $i \in [j-1]$, there is an edge between the colors $X \in \P_{i}$ and 
$Y \in \P_{j}$ if they share at least one variable.
Observe that if any two colors $X_a$ and $X_b$ of partition $\P_j$
are reachable from each other in $G_j$,
then, they should be in the same neighborhood.
As reachability is an equivalence relation, {\em the neighborhoods are equivalence classes of colors}. 
%In a $\d$-distance circuit,
%given a color $X_a$ of partition $\P_j$,
%atmost $\d$ colors of the partition $\P_j$
%are reachable from the color $X_a$, in graph $G_j$.

Moreover, observe that for any two variables $x_a$ and $x_b$,
if their respective colors in partition $\P_j$, $\color_j(x_a)$ and $\color_j(x_b)$
are reachable from each other in $G_j$
then their respective colors in partition $\P_{j+1}$, 
$\color_{j+1}(x_a)$ and $\color_{j+1}(x_b)$
are also reachable from each other in $G_{j+1}$. Hence,
\begin{observation}
\label{obs:zeroBelow}
If at some partition,
the variables $x_a$ and $x_b$ are in the same neighborhood,
then, they will be in the same neighborhood in all of the lower partitions.
I.e.\ $\nbd_j(x_a) = \nbd_j(x_b) \implies \nbd_i(x_a) = \nbd_i(x_b), \forall i \ge j$.
\end{observation}
In other words, at the level of the variables, the neighborhoods in the upper partitions are {\em refinements} of the neighborhoods in the lower partitions. 

%For a set of variables $S$, let $\abs{\nbd_j(S)}$ be the number of colors in $\nbd_j(S)$ and $\norm{\nbd_j(S)}$ be the number of neighborhoods in $\nbd_j(S)$.
%In a $\d$-distance circuit,
%\begin{observation}\label{obs:noOfColors}
%For any set $S$ of variables, in any partition $\P_j$, the number of colors in the neighborhood %of $S$, $\abs{\nbd_j(S)} \le \d \norm{\nbd_j(S)} \le \d\abs{S}$.
%\end{observation}

%Take the partitions $\P_j$ and $\P_{j-1}$ for $2 \le j \le n$.
%Consider the colors in the partitions $\P_j$ and $\P_{j-1}$ as vertices of a bipartite graph. %Put an edge between the colors $X \in \P_j$ and $Y \in \P_{j - 1}$ if they share at least one variable.
%Let $X_a \sim X_b~(X_a, X_b \in \P_j)$ if $X_a$ is reachable from $X_b$ in the graph.
%Hence,
%$\sim$ is an equivalence relation. Observe that $X_a \sim X_b$ iff they are in each others' %neighborhood.
%
%Thus, neighborhoods are equivalence classes of colors.
%The partition $\P_{j-1}$ can be replaced with any upper partition $\P_i$ of the partition $\P_j$ in the above argument.

We now claim that any subcircuit of a $\d$-distance circuit $\Dx$, 
also has $\d$-distance. 
%Hence,
\begin{observation}[Subcircuit of $\Dx$]
\label{obs:subCktHasDeltaDist}
Let $E\in \H_j(\F)[\o{x}]$ be a subcircuit of a $\d$-distance circuit $\Dx$, obtained by replacing an arbitrary set of linear factors in each coordinate of $\Dx$ with $1$,
and restricting the circuit to the coordinates $1$ to $j$.
Then $E$ is also a $\d$-distance circuit.
\end{observation}
\begin{proof}
%The proof is in two steps.
A linear factor $ b_{i_0} + \sum_r b_{i_r} x_{i_r} $
which is replaced with $1$,
can be viewed as
$ 1 + \sum_r 0 \cdot x_{i_r} $.
Such a subcircuit induces the same partition sequence
%$(\P_1, \P_2, \dots, \P_k)$, 
as circuit $D$.
When we restrict the circuit to the coordinates $1$ to $j$, we get a subsequence of
this partition sequence. Clearly, the subsequence also has $\d$-distance.
%Now, a subsequence of a partition sequence with $\d$-distance also has $\d$-distance.
\end{proof}
Since there exists a set of colors (linear factors in the circuit $C$) in $\P_i~(i \in [j])$ that exactly contain the variables of one neighborhood,
$\nbd_j(X_a)$, % in $\P_j, \forall j \in \n$,
we can define the following subcircuit of $\Dx$.  

\begin{definition}[$\tower_j({\X})$]
For a neighborhood $\X$ in partition $\P_j$, we define a $\tower_j({\X})$ as a \p over
%over $k_1 + k_2 + \dots + k_j$ 
$\H_j(\F)$, such that its 
%the coordinates corresponding to any partition $\P_i$ ($i \leq j$) 
$i$-th coordinate ($i \leq j$) 
is the product of exactly those linear factors
(in the $i$-th product gate of the circuit $C$)
that contain the variables of the neighborhood $\X$.
\end{definition}
%$\tower_j(\X)$ forms a subcircuit of the original circuit $\Dx$.
%The product gates correspond to partitions $\lis{\P}{,}{j}$ and the linear terms are colors in $\nbd_j(X_a)$.
%Let us call this subcircuit corresponding to $\nbd_j(X_a)$ as $\tower_j(X_a)$.
%Similarly, 
We can define a tower over 
a union of neighborhoods $ (\cup_{i=1}^r \X_i )$ in partition $\P_j$ as
$\tower_j(\cup_{i=1}^r \X_i) := \tower_j(\X_1) \star \tower_j(\X_2) \star \dotsm  \star \tower_j(\X_r)$. For any such tower $E = \tower_j(\cup_{i=1}^r \X_i) $, $\nbd_j(E)$ will denote the union of neighborhoods $(\cup_{i=1}^r \X_i)$.
%The definition of a tower can be extended to a tower over a set of neighborhoods:
For a set of variables $S$,
$\tower_j(S)$, is the tower over $\nbd_j(S)$.
%For a color $X_a$ in partition $\P_j$, $\tower_j(X_a)$ is defined 
%as $\tower_j(\nbd_j(X_a))$.
For a neighborhood $\X$ in partition $\P_j$,
the variables in $\tower_j(\X)$ are the variables in $\X$, 
denoted by both $\variables(\tower_j(\X))$ and $\variables(\X)$.
Observe that the towers over any two neighborhoods in partition $\P_j$ are polynomials over
a disjoint set of variables.

%Like $\tower_j(X_a)$ is a `tower' over $\nbd_j(X_a)$,
%we can define $\tower_j(x_i)$ for a variable $x_i$, $\tower_j(\{X_{a_i}\}_{i \in {\mathcal I}})$ for a set of colors $\{X_{a_i}\}_{i \in {\mathcal I}}$ in partition $\P_j$ and $\tower_j(\{x_{i}\}_{i \in {\mathcal I'}})$ for a set of variables $\{x_{i}\}_{i \in {\mathcal I'}}$ in  partition $\P_j$ as towers over $\nbd_j(x_i)$, $\nbd_j(\{X_{a_i}\}_{i \in {\mathcal I}})$ and $\nbd_j(\{x_{i}\}_{i \in {\mathcal I'}})$ respectively.

The following observation says that the coefficient of a monomial in a product of towers is equal to the product of its `support coefficients' in the individual towers.
\begin{observation}
\label{obs:CoeffMulInTowers}
Let $R$ and $T$ be two sets of variables coming from two \emph{disjoint sets }of neighborhoods of partition $\P_j$.
%Let $\tower_j(R)$ and $\tower_j(T)$ be two different towers over the partition $\P_j$ with disjoint set of variables.
Then, the coefficients of monomial $S \subseteq R \cup T$ in $\tower_j(R) \star \tower_j(T)$ is given by $\coeff_{\tower_j(R) \star \tower_j(T)} (S)
 = \coeff_{\tower_j(R)} (S \cap R) \star \coeff_{\tower_j(T)} (S \cap T)$.
\end{observation}

%\begin{observation}\label{towerRefinement}
%A subcircuit at level $j$, $\tower_j$ may consist of more than one subcircuits at level $j-1$.
%\[
%\abs{
%\left\{\tower_{j-1}(x_i) \mid x_i\in \tower_j
%\right\}
%} \ge 1.
%\]
%In other words, the neighborhoods in the upper partitions are refinements of the neighborhoods %in the lower partitions.
%\end{observation}

%\vspace{-0.25cm}
\subsection{True coefficients}
%Like \cite{ASS13}, w
%Since the $\Pi\Sigma$ circuit may have high-support monomials,
We will be proving low-support concentration in some special subcircuits $\Epx$ (see Observation \ref{obs:subCktHasDeltaDist}) of $\Dpx$ (the shifted polynomial).
This would eventually prove low-support concentration of $\Dpx$.
%
%
%
%We will be proving low-support concentration for low-support coefficients.
%and then it would imply \lc for $(>\ell)$-support coefficients.
%To prove \lc for $\coeff_{D'}(x_S), \abs{S} = \ell$,
%we will prove \lc for $\coeff_{D'}(x_S)$ in the subcircuit $D'_S(\o{x}) = \tower_k(S)$ (roughly).
The coefficient $\coeff_{E'}(x_S)$ of the monomial $x_S$
in the subcircuit $\Epx$ is given by $\coeff_{D'}(x_S) = \const_{E'} \star \coeff_{E'}(x_S)$, where
$\const_{E'}$ is the product of the constant parts of all linear factors not in $\Epx$. 
%$F'$ is the polynomial obtained by setting the linear terms in $\Epx$ to 1.
Thus, it is enough to prove \lc within $\Epx$ since any dependency in $\Epx$ translates to a dependency in $\Dpx$ by multiplying throughout with a constant.

All the monomials which participate in a dependency, and hence, all the monomials occurring in $\Epx$ should have {\em true coefficients}.
%I.e.\ the monomial should be present in $\DpS$ in every coordinate, or not at all.
I.e.\ each monomial should have a coefficient in $\Epx$ similar to that in $\Dpx$.
I.e.\ $\const_{E'} \star \coeff_{E'}(x_T) = \mathbf{0}$ or $\coeff_{D'}(x_T)$, for all monomials $T$.

In the next subsection, we will
study a few properties of the subcircuit $\Epx$, for which we prove low-support concentration.
%try to shortlist the properties required by the subcircuit $\Epx$.
%One requirement from $\DS$ is: Since \lc is proven after shifting, all the variables in a linear term contribute to the constant of that linear term. 
%So, if one variable from a linear term $\ell_{ij}$ is picked to be in $\DS$, all the variables in the linear term $\ell_{ij}$ should be picked to be in $\DS$.
%%So, the polynomial $\DS$ includes the whole linear terms that include the variables in $S$.
%
%We will now study another requirement from the \p $\DS$.

%\vspace{-0.25cm}
\subsection{Phases}\label{phases}
%The subcircuit $\Epx$, for which we prove low-support concentration
%would have 
Since true coefficients have to be used, $\Epx$ has to be (roughly) a tower over a (small) set of neighborhoods.
Thus, though the lowest partition in $\Epx$ will have a few linear factors,
the higher partitions may have $O(n)$ linear factors.
This blows up the support size of the monomials in $\Epx$.
%One way to proceed would be to directly show that
%there exists dependencies such that when lifted,
%coefficients of such high-support monomials is $0$.
%We use another method.
We intend to show (by induction) that these coefficients of the high-support monomials are $\ell$-concentrated.
Note that these high-support monomials come only from the upper partitions.

%\begin{observation}\label{obs:zeroBelow}
%Because of the refinement of the neighborhoods (Observation \ref{towerRefinement}),
%if at some partition,
%the variables $x_a$ and $x_b$ are in the same neighborhood,
%then, they will be in the same neighborhood in all of the lower partitions.
%I.e.\ $\nbd_j(x_a) = \nbd_j(x_b) \implies \nbd_i(x_a) = \nbd_i(x_b), \forall i \ge j$.
%\end{observation}

We use the refinement property of neighborhoods (Observation~\ref{obs:zeroBelow}) 
to categorize the monomials of the \p $\Dx$ into $k$ phases.
Roughly speaking, the \emph{phase} of a monomial $S$ is the lowest partition from which $S$
can possibly be generated.
If $(\d + 1)$ variables of $S$ belong to the same neighborhood of partition $\P_{j+1}$,
then $S$ cannot be generated from this partition, or the partitions below it, as we will soon see.
% because of Observation \ref{obs:zeroBelow}.
This motivates us to define the phase of a monomial as:

\begin{definition}[Phase-$j$]\label{def:phase}
A monomial $S$ is in phase-$j$ if the partition $\P_j$ is the lowest partition with $\le \d$ of its support variables in each of its neighborhoods.
\end{definition}
%\vspace{-0.1cm}
A phase-$ j$ monomial can be characterized by the presence of a \emph{pivot tuple}.
\begin{observation}[Pivot tuple]
\label{obs:pivot}
Let $j \in [k-1]$.
A monomial $S$ is in the \emph{$j$-th phase}
iff
\begin{enumerate}
\item There is no neighborhood of the upper partitions $\{\P_i\}_{i = 1}^j$ that contributes more than $\d$ variables to the monomial $S$, and,
\item there exist $(\d+1)$ variables in its support
(called a \emph{pivot tuple}) %, $\pivot(S)$)
that are in the same neighborhood of the partition $\P_{j+1}$.
This pivot tuple is in the same neighborhood of all the lower partitions $\{\P_i\}_{i = j + 1}^ k$.
\end{enumerate}
%$\forall i \le j$,% it is not in the $i$-th phase.
%I.e.\ 
A monomial is in the $k$-th phase if it is  not in any of the previous phases.
\emph{(Proof is by applying Observation \ref{obs:zeroBelow} on Definition \ref{def:phase}.)}
%From the definition, it is clear that a monomial $S$ in phase-$j$ has no neighborhood in the upper partitions $\P_i, \forall i \le j$, that contributes more than $\d$ variables to the monomial $S$.
\end{observation}
\begin{observation}
A monomial $S$ belongs to exactly one of the $k$ phases, denoted by $\phase(S)$.
\end{observation}
%\vspace{-0.1cm}
\begin{observation}
\label{obs:phaseOfPivot}
A monomial $S$ may have more than one pivot tuples. Observe that $\phase(R) = \phase(S)$, for all pivots $R$ of a monomial $S$.
\end{observation}
%\vspace{-0.1cm}
Since every superset of a monomial $S$ contains its pivot tuple, $\pivot (S)$, %the
%following observation can be made.
%\vspace{-0.1cm}
\begin{observation}\label{obs:supSetPhase}
For all supersets $T \supseteq S$, $\phase(T) \le \phase(S)$.
\end{observation}

In a $\d$-distance circuit, a neighborhood has at most $\d$ colors.
Hence, for any $\phaseh j$ monomial, at least two of its $(\d+1)$ pivot variables, have the same color in the lower partitions $\{\P_i\}_{i = j+1}^k$. They come from the same linear factor of the circuit $C$. Thus, %Hence, the coefficient of this monomial has $0$ in the coordinates $\{j+1, j + 2, \dots, k\}$.

%Observe that a monomial in the $j$-th phase has $(\d+1)$ variables in a single neighborhood in all of the lower partitions, i.e.\ $\P_{j+1}, \P_{j+2}, \dots, \P_{k}$.

\begin{observation}
\label{obs:chopping}
If a monomial $S$ is in the $j$-th phase, then, its coefficient vector is zero in the coordinates $\{ j+1, j+2, \dots, k\}$.
%~\coeff_D(x_S)(j+1, j+2, \dots, k) = ({0}, {0}, \dots, {0})$.
\end{observation}

Since the coordinates $\{j+1, j+2, \dots, k\}$
in the coefficients of phase-$(\le j)$ monomials are $0$,
they can be ignored when
the linear dependencies among the coefficients of phase-$(\le j)$ monomials are studied.
For phase-$j$,
we discuss the linear dependencies among coefficients of phase-$(\le j)$ monomials,
i.e.\ vectors limited to the $(\le j)$-th coordinates.
Coefficients of phase-$j$ monomials will be called phase-$j$ coefficients.
%\[
%\DS = \tower(S).
%\] and has $\phase(S)$ coordinates.
%We let $\nbd(T) := \nbd_{\phase(T)}(T)$ and $\tower(T) := \tower_{\phase(T)}(T)$.

%\footnote{Should we associate a single pivot with a set $S$? Or should we associate a set of pivots with a monomial of phase-$j$? Then, $\Pivot_j = \bigcup_{S \in \phaseh  j} \pivots(S)$. 
%}
%-------Motivate this!--------
%\begin{observation}\label{supset}
%For all supersets $T \supseteq S$, $\phase(T) \le \phase(S)$.
%\end{observation}
%\begin{observation}\label{phaseSubset}\footnote{Change? Is it useful now?}
%Let a monomial $S$ be in the $j$-th phase, such that the number of $\P_j$ neighborhoods it spans across is $b$. $\abs{\nbd(S)} = b$. Let $\abs{\nbd(\pivot(S))} = a$. Then, for every neighborhood size $s: a \le s \le  b$, the set $S$ has a subset $T$ in the $j$-th phase such that $\abs{\nbd(T)} = s$. $\abs{\nbd(S) \setminus \nbd(T)} = b - s$.
%\end{observation}
%\begin{proof}
% $T$ is obtained by including variables from $\nbd(\pivot(S))$ and all variables from other neighborhoods.
%\end{proof}

%\vspace{-0.1cm}
\subsection{Partial derivatives}
All the monomials in phase-$ j$ are identified by the existence of at least one pivot tuple.
Consider a phase-$ j$ monomial $S$ with a pivot $R$ of phase-$ j$ (Observation \ref{obs:phaseOfPivot}).
%We intend to show that
%the coefficient of $S$ depends on certain coefficients of phase-$(\le j)$.
%%and $(<\ell)$-support monomials of phase-$ j$.
%%One way to achieve this is to 
%We show that it depends on %such 
%the coefficients of the `smaller' monomials that include $R$.
We will prove \lc 
(linear dependence on $(\le \ell)$-support coefficients)
among the coefficients of all the monomials that include $R$.
Since all the monomials now include $R$, the $(> j)$-th coordinates in the \p $\Dpx$ are $0$ (Observations \ref{obs:phaseOfPivot}, \ref{obs:supSetPhase}, \ref{obs:chopping}).
The polynomial could be restricted to coordinates $1$ to $j$.
We prove \lc after taking \emph{partial derivative} with respect to $x_R$.
This is encapsulated in Observation \ref{obs:coeffMul}.

%, and hence, are in phase-$(\le j)$ (Observation \ref{obs:supSetPhase}). This is done by taking \emph{partial derivative }of $\Dpx$ with respect to $x_R$.

In other words: In the linear factors of $C$ corresponding to the partitions $\{\P_i\}_{i = 1}^j$ that include the variables of $R$,
we pick the monomial $x_R$ and its coefficient.
Then, we prove low-support concentration in the remaining polynomial.
\subsection{\lc}
We now come to the main subsection of the paper. Let $\ell_0 = \log_2(k + 1), \ell = \d + 1 + \ell_0$ and $\Dpx = D(\o{x} + \phi_{\delta \ell_0}(\o{t}))$, where $\Dx \in \H_k(\F)[\o{x}]$ is a $\d$-distance circuit and
$\phi_{\delta \ell_0}$ (defined in Subsection \ref{subsec:kronecker})
is a univariate map
which separates all ($\leq \delta \ell_0 $) support monomials.
We prove \lc in the \p $\Dpx$ (Lemma~\ref{lem:dDistlConc}), because of which, a $n^{O(\d \log k)}$ time hitting-set exists for the $\d$-distance depth-$3$ \ml circuit (Theorem~\ref{thm:dDistHS}).

%\vspace{-0.1cm}
\subsubsection{\lc in each phase}

As already stated in Subsection \ref{phases},
proof of \lc in $\Dpx$ is by induction on the phases.
The monomials of each phase are progressively shown to be $\ell$-concentrated:
In this subsection, when we say that a phase-$i$ monomial is $\ell$-concentrated,
we mean that its coefficient is linearly dependent on coefficients of $(<\ell)$-support monomials of phase-$(\leq i)$.
The following lemma shows one step of the induction.

\begin{lemma}[Phase-$(<j)$ to Phase-$j$]\label{lem:j-1Impliesj}% j-1 implies j
If the monomials in $\Dpx$ of phase-$(<j)$ are $\ell$-concentrated,
% for $\ell = \d + 1 + \log_2(k+1)$, 
then the monomials of phase-$ j$ are $\ell$-concentrated. $\ell$-Concentration exists in phase-$ 1$ without any preconditions. 
\end{lemma}

\begin{proof}
Since phase-$0$ does not exist, the monomials of phase-$0$ are vacuously $\ell$-concentrated.
Now, when the monomials of phase-$(<j)$ are $\ell$-concentrated, 
% for $\ell := \d + 1 + \log(k+1)$,
Lemma~\ref{ConcInMonomialsIncludingR} proves \lc in all phase-$ j$ monomials that include the phase-$j$ pivot $R$. Every monomial of phase-$ j$ has such a pivot from Observation~\ref{obs:pivot}.
% from Observation~\ref{smallPivot}.
\end{proof}

The following observation about $\Dpx$ will be used in the next lemma.
%\vspace{-0.1cm}
\begin{observation}[$\DpR, \EpR$ and Decomposition of coefficients]
\label{obs:coeffMul}
%For any monomial $x_R \subseteq \n$,
For any $R \subseteq \n$,
we define $\DpRx \in \H_{k}(\F(t))[\o{x}]$ to be the polynomial,
obtained from $\Dpx$,
where each coordinate is the product of those linear factors, in the corresponding coordinate of $\Dpx$, that include the variables of $R$.
Let $\EpRx\in \H_{k}(\F(t))[\o{x}]$ be the polynomial,
obtained from $\Dpx$,
where each coordinate is the product of those linear factors, in the corresponding coordinate of $\Dpx$, that do not include the variables of $R$.
%\footnote{$\coeff_{\DpR}(x_R) \star \EpRx = \frac{\partial}{\partial x_R} \left(\Dpx\right)$.
%}.
I.e.\ $\Dpx = \DpRx \star \EpRx$.
Then, for all supersets $S$ of $R$, $\coeff_{D'}(x_S) = \coeff_{\DpR} (x_R) \star \coeff_{\EpR} (x_{S \setminus R})$.
\end{observation}
\begin{proof}
The observation is easy for the coordinates where the linear factors that support $R$ and $S \setminus R$ are disjoint.
In the coordinates where
at least one variable, say $x_a$, of $S \setminus R$
is included in the linear factors in $\DpRx$,
$\coeff_{\EpR} (x_{S \setminus R})$ is zero in these coordinates
because the variable $x_a \in S \setminus R$
is present in only one linear factor,
which is not in $\EpR$.
Whereas $\coeff_{D'}(x_S)$ is zero in these coordinates
because at least two variables of $S$ are present in the same linear factor in $\DpR$.
\end{proof}

\begin{lemma}[Concentration in monomials that include $R$]\label{ConcInMonomialsIncludingR}
Let $R$ be a phase-$j$ monomial such that $\abs{R} = \d + 1$. If \lc at phase-$(<j)$ exists, then the monomials that include $R$ are $\ell$-concentrated.
\end{lemma}
\begin{proof}
We show that the coefficients, whose support include $R$, are linearly dependent on: The low-support coefficients, whose support include $R$, and phase-$(<j)$ coefficients.

Consider the polynomial $\EpRx$ of $\Dpx$, restricted to coordinates $1$ to $j$.
%\footnote{We abuse the notation $D', \EpR$ and $\DpR$ and use them for the corresponding polynomials restricted to coordinates $1$ to $j$.}.
%$(1, 2, \dots, j)$\footnote{We are assuming here, without loss of generality, that the vector $\coeff_{\DpR}(x_R)$ is invertible. If there are any zero-coordinates, then we can consider the smaller Hadamard algebra.}.
Lemma~\ref{lem:subcircuit}, together with Observation \ref{obs:subCktHasDeltaDist}, proves \lzc for the monomials in $\EpRx$ modulo the coefficients in $\phasesh (<j)$. I.e.\ $\forall $ monomial $S' $ in $ \EpRx$,
\[
\coeff_{\EpR} (x_{S'}) 
= \sum_{\substack{T': \abs{T'} < \ell_0,\\{T' \in \phaseh  j \mbox{ of } \EpR}}} 
\alpha_{T'} \coeff_{\EpR} (x_{T'}) 
+ \sum_{\substack{{T'} \in \phaseh (<j) \mbox{ of } \EpR }} 
\alpha_{T'} \coeff_{\EpR} (x_{T'}) ,
\]
where $\alpha_{T} \in \F(t)$ and the monomials are from $\EpRx$.
Since the supports of $R$ and $S'$ are disjoint, 
from Observation \ref{obs:coeffMul}, $\coeff_{D'}(x_{S'\cup R}) = \coeff_{\DpR} (x_R) \star \coeff_{\EpR} (x_{S'}), ~\forall $ monomial $S' $ in $ \EpRx$.

Therefore, multiplying throughout with $\coeff_{\DpR} (x_R) $,
\[
\coeff_{D'} (x_{S' \cup R}) 
= \sum_{\substack{T': \abs{T'} < \ell_0,\\{T' \in \phaseh  j \mbox{ of } \EpR}}}
 \alpha_{T'} \coeff_{D'} (x_{T' \cup R})
  + \sum_{\substack{{T'} \in \phaseh (<j) \mbox{ of } \EpR}}
   \alpha_{T'} \coeff_{D'} (x_{T'\cup R}).
\]

%The phase-$ (\le j)$ coefficients remain in the same phase even after multiplication with $R$, since $R$ is a phase-$ j$ monomial and $tower_j(R)$ and $\EpRx$ are disjoint.
%Since $T' \in \tower(S')$, $T'\cup R \in \tower(S' \cup R)$.

Let $S = S' \cup R $ and $T = T' \cup R$.
From Observation~\ref{obs:supSetPhase},
even after including $R$, a phase-$(<j)$ monomial remains a phase-$(<j)$ monomial
and a phase-$j$ monomial remains a phase-$( \le j)$ monomial. $\abs{R} = \d + 1$. Rewriting gives,
\[
\coeff_{D'} (x_{S})
 = \sum_{\substack{T: \abs{T} < \ell_0 + \d + 1,\\{T \in \phaseh  (\le j) \mbox{ of } D'}}}
  \alpha_{T'} \coeff_{D'} (x_{T})
   + \sum_{\substack{{T} \in \phaseh (<j) \mbox{ of } D'}}
    \alpha_{T'} \coeff_{D'} (x_{T}).
\]

$\coeff_{D'} (x_{S})$ is thus linearly dependent on: $(<\ell)$ support coefficients, and the coefficients in the upper phases. The upper phase monomials are already $\ell$-concentrated. Hence, $\coeff_{D'} (x_{S})$  is also $\ell$-concentrated.

This is proven for every monomial $S$ with $R$ as its subset.
\end{proof}

%\vspace{-0.1cm}
\subsubsection{$\ell_0$-concentration in the ``last" phase}
Lemma~\ref{ConcInMonomialsIncludingR} needed $\ell_0$-concentration
of $\EpR$ modulo its phase-$(<j)$ coefficients. 
Similar to $\EpR$, we can define a polynomial $\ER$ of $D$ (Observation~\ref{obs:coeffMul}).
It is easy to see that $\EpR = \ER(\o{x}+ \phi_{\delta \ell_0}(\o{t}))$.
By Observation \ref{obs:subCktHasDeltaDist},
$\ER$ is also a $\d$-distance circuit.
We will actually give a general result that 
for any $\d$-distance circuit $E(\o{x}) \in \H_{j}[\o{x}]$,
its shifted version $E' := E(\o{x}+ \phi_{\delta \ell_0}(\o{t}))$ 
has $\ell_0$-concentration modulo its phase-$(<j)$ coefficients. 
Here, phases for coefficients, neighborhoods in the partitions
and towers over the neighborhoods, are defined for $E'$, similar to $D'$.
Note that $j$ is the last phase of $E'$,
as it is over $\H_{j}(\F(t))$.
For any $E'(\o{x}) \in \H_{j}(\F(t))[\o{x}]$, let $V_{E'}$ 
denote the space spanned by the phase-$(<j)$
coefficients of $E'$.

%Similar to $D'$,
%we define the neighborhoods in the partitions, the towers on the neighborhoods and the $j$ %phases for the coefficients of $E'$.   

%We will now show low-support concentration (modulo some space)
% in certain subcircuits 
%of $\Dpx$. 
%These subcircuits %will not be completely arbitrary but 
%correspond to
%a tower over some set of neighborhoods in some partition $\P_j$. 
%Let us denote this subcircuit by $E'(\o{x})$.
%Observe that $E'$ can be viewed as a shifted version of
%a subcircuit $E$ of $D$, i.e.\

%$E'(\o{x}) = E(\o{x} + \phi_{\d\ell_0}(\o{t}))$.
%Here, $E$ is a subcircuit of $D$,
%corresponding to a tower over some set of neighborhoods in partition $\P_j$. 
%Note that the polynomial $E'(\o{x})$ lies in $\H_{j}(\F(t))[\o{x}]$,
%where $k' := \sum_{i=1}^j k_i$.
%Similar to $D'$, we define the phases for the coefficients
%of $E'$. 

%s.
%Here, $E$ is a subcircuit of $D$.
%, corresponding to a tower over some set of neighborhoods in partition $\P_j$. 
%Note that the polynomial $E'(\o{x})$ lies in $\H_{j}(\F(t))[\o{x}]$.

%By Observation \ref{obs:subCktHasDeltaDist},
%$E'(\o{x})$ is a $\d$-distance circuit

%Let $\ell_0 = \log (k +1)$.
%We will actually show $\ell_0$-concentration in $E'$
%modulo its phase-$(<j)$ coefficients. 

%Now, we prove $\ell_0$-concentration
%for any tower over a set of neighborhoods in partition 
%$\P_j$, without the restriction 
%on the number of neighborhoods.
The proof is in two parts.
Let $\norm{\nbd_j(E')}$ be the number of neighborhoods
in $E'$ in its partition $\P_j$.
The first part is to show the result for smaller circuits $E'$,
%i.e.\ the number of neighborhoods in its partition $\P_j$, 
when $\norm{\nbd_j(E')} \leq \ell_0$.
This is done in Lemma \ref{lem:subcircuitl0}.
The second part is to prove it for larger circuits using induction. 
This is done in Lemma \ref{lem:subcircuit}.

%\vspace{-0.1cm}
\begin{lemma}[Small-neighborhood $\ell_0$-concentration]
\label{lem:subcircuitl0}
Let $\Ex \in \H_j(\F)[\o{x}]$ be a $\d$-distance circuit
such that $\norm{\nbd_j(E')} \leq \ell_0$.
Then $E'(\o{x}) = E(\o{x} + \phi_{\d\ell_0}(\o{t}))$ 
%be its shifted version.
%Let $V_{E'}$ be the subspace spanned by the coefficients of the phase-$(<j)$ monomials of $\Epx$.
%Then, $\Epx$ 
has $\ell_0$-concentration modulo $V_{E'}$. \emph{(Proof in Appendix \ref{app:subcircuitl0}.)}
\end{lemma}
% The proof uses a series of lemmas, all proved in the appendix. Out of these, Lemma \ref{lem:invertible} (stated next) was directly invoked in the proof of Lemma \ref{lem:subcircuitl0}.

%We write the second part first.
%\vspace{-0.1cm}
\begin{lemma}[Subcircuit-$\ell_0$-concentration]
\label{lem:subcircuit}
Let $\Fx \in \H_j(\F)[\o{x}]$ be a $\d$-distance circuit.
Then, $F'(\o{x}) = F(\o{x} + \phi_{\d\ell_0}(\o{t}))$ 
%be its shifted version.
%Let $V_{F'}$ be the subspace spanned by the coefficients of the phase-$(<j)$ monomials of $\Fpx$.
%Then, $\Fpx$ 
has $\ell_0$-concentration modulo $V_{F'}$.
\emph{(Proof in Appendix \ref{app:dDist}.)}
\end{lemma}

%\vspace{-0.1cm}
\subsubsection{Proving Theorem \ref{thm:dDistHS}}

From Lemma~\ref{lem:j-1Impliesj}, we directly get Lemma \ref{lem:dDistlConc}:
\begin{lemma}\label{lem:dDistlConc}
If $D \in \H_k(\F)[\o{x}]$ is a $\d$-distance circuit, then $D'$ has $\ell$-concentration.
%, where $\ell_0 := \log( k + 1)$ and $\ell :=  \d + 1 + \log (k +1)$.
\end{lemma}

We now prove the main theorem. 
If \lc exists, then a hitting-set exists.% cannot be far behind!
%If \lc exists, then a hitting-set cannot be far behind!
\begin{proof}[Proof of Theorem~\ref{thm:dDistHS}]
We can write $C(\o{x}) = \o{a}^T D(\o{x})$ for some $\o{a} \in \F^k$ and $\Dx \in \H_k(\F)[\o{x}]$.
From Lemma~\ref{lem:dDistlConc}, $D(\o{x} + \phi_{\delta \ell_0}(\o{t}))$ has $\ell$-concentration.
%where $\ell_0:= \log(k+1)$ and $\ell:= \d + 1 + \log(k+1)$. 
Now, the map $\phi_{\delta \ell_0}$ separates
all the monomials of support $\leq \delta \ell_0$. The number of such monomials is $n^{O(\delta \ell_0)}$.
Hence from Lemma~\ref{lem:kronecker} (Appendix \ref{app:dDist}), $\phi_{\delta \ell_0}$ can be generated by trying $N$-many monomial maps which have degree $\leq N \log N$, where $N := n^{O(\delta \ell_0)}$. 
Now, from Lemma~\ref{lem:finalHS} we directly get a $n^{O(\delta \log k )}$-time hitting-set. 
\end{proof}

%%%%%%%%%%%%%%%%%%%%%%%%%%%%%%%%%%%%%%%%%%%%%%%%%%%%%%%%%%
%% Section Base sets 
%%%%%%%%%%%%%%%%%%%%%%%%%%%%%%%%%%%%%%%%%%%%%%%%%%%%%%%%%%%%%%%%%%%%%

%\vspace{-0.25cm}
\section{Base sets with $\delta$ distance: Theorem~\ref{thm:baseSetsHS}}
\label{sec:baseSets}
In this section we further generalize the class 
of polynomials, for which we can give an efficient test, 
beyond low-distance. 
Basically, it is enough to have low-distance ``projections". 

\begin{definition}
A multilinear depth-$3$ circuit is said to have 
$m$-base-sets-$\delta$-distance
if there is a partition of the 
variable set $\o{x}$
into base sets $\{\lis{\o{x}}{,}{m}\}$ such that 
for any $i \in [m]$, 
restriction of $C$ on the $i$-{th} base set, 
i.e.\ $C|_{(\o{x}_j =\mathbf{0} \; \forall j \neq i)}$ has $\delta$-distance. 
\end{definition}

As discussed in the previous section, for a depth-$3$
circuit $C$ we can write $C(\o{x}) = \o{a}^T \cdot D(\o{x})$,
where $\o{a} \in \H_k(\F)$ and $D(\o{x}) \in \H_k(\F)[\o{x}]$.
We say a polynomial $D(\o{x}) \in \H_k(\F)[\o{x}]$ 
has $m$-base-sets-$\delta$-distance if the corresponding 
circuit $C$ has $m$-base-sets-$\delta$-distance.  
We will show that such a polynomial $\Dx$ will also have some
appropriately low-support concentration after an efficient shift (Lemma~\ref{lem:baseSetConc}, Appendix~\ref{app:baseSets}). 

From Lemma~\ref{lem:dDistlConc}, we know that if a polynomial 
$\Dx \in \H_k(\F)[\o{x}]$ has $\delta$-distance
then $D(\o{x}+\phi_{\delta \ell_0}(\o{t}))$ has $\ell$-concentration. 
%where $\ell_0 := \log(k+1)$ and $\ell := 2\delta + \ell_0 $. 
The basic idea for a polynomial $\Dx$ having
 $m$-base-sets-$\delta$-distance, is to use a different
shift variable for each base set. Hence, it is necessary that 
the base sets are known. Except this knowledge, the test is blackbox.
The basic argument is to view the polynomial $D(\o{x})$ as
a polynomial over the variable set $\o{x}_m$ and whose coefficients lie
in $\H_k(\F)[\o{x}_1,\o{x}_2, \dots, \o{x}_{m-1}]$. 
From this perspective, it has $\delta$-distance and hence we can achieve
low-support concentration. Further, we argue low-support concentration in
its coefficients, which themselves have $(m-1)$-base-sets-$\delta$-distance (Lemma~\ref{lem:recursive}). 
The proof is by induction on the number of base sets.

%As mentioned earlier, to achieve low-support concentration
 %we will use a different shift 
%variable for each base set. 
Let $\o{t}_i$ be the set of shift variables for $\o{x}_i$ for 
any $i \in [m]$ ($\abs{\o{t}_i} = \abs{\o{x}_i}$).
Let $I_i$ denote the set of indices corresponding to the variables in 
$\o{x}_i$, for all $ i \in [d]$.
We define our shifting map $\phim$ as follows:
$\forall i \in m$, $\phim \colon \o{t}_i \mapsto \{t_i^j\}_{j=0}^{\infty}$
such that for any two sets $S,T \subseteq I_{i} $ with
$\abs{S},\abs{T} \leq \delta \ell_0$, 
$\phim(t_S) \neq \phim(t_T)$ for all $i \in [m]$.
%$\phim(\o{t}_i)$ will denote the tuple of the images 
%of $\phim$ on $\o{t}_i$ which will consists of monomials
%in variable $t_i$.

\begin{proof}[Proof of Theorem~\ref{thm:baseSetsHS}]
Let $C(\o{x}) = \o{a}^T \cdot D(\o{x})$ for some $\o{a} \in \H_k(\F)$
and $\Dx \in \H_k(\F)[\o{x}]$.
Lemma~\ref{lem:baseSetConc} (Appendix~\ref{app:baseSets}) shows that $D(\o{x} + \phim(\o{t}))$
has $(m(\ell - 1)+1)$-concentration.
%, where $\ell = \log( k+1) + 2\delta $. 
%It means if $C \neq 0$ then at least one of its coefficients
%with support $\leq m(\ell - 1)$ is nonzero. 
%There are $O(n^{m(\ell - 1)})$ such coefficients. 
Hence, from Lemma~\ref{lem:hsFromlConc}, $C(\o{x} + \phim(\o{t}))$ has 
$n^{O(m \ell )}$-time hitting-set. 
Moreover, each evaluation of $C(\o{x} + \phim(\o{t}))$
 is a polynomial in $\{t_1, t_2, \dots, t_m\}$. 
Let us say the individual degree bound on this polynomial is $d$.
Then the time taken to compute this polynomial would be 
proportional to the number of monomials in it, i.e.\ $(d+1)^m$.

For the degree bound we must look at the map $\phim$.
The map $\phim$ separates all $\delta \ell_0 $ support monomials.
There are $n^{O(\delta \ell_0)}$ such monomials. 
From Lemma~\ref{lem:kronecker} (Appendix~\ref{app:dDist}), we need to 
try $N$ maps each with highest degree $N \log N$ to
get the desired map,
%where $N = n^{2\delta \ell_0 +1}$.
where $N = n^{O(\delta \ell_0)}$.
Hence, the total complexity is $n^{O(m\ell)}N^{O(m)} = n^{O(m \delta \log k)}$.
\end{proof}

%Now, we remark that when $m = O(\log n)$ and $\delta = O(\log n)$,
%the time complexity is $2^{O(\log^4 n)}$.

%**********************************************************
%section - ROABP
%***********************************************************

%\vspace{-0.25cm}
\section{Sparse-Invertible Width-$w$ ROABP: Theorem~\ref{thm:ROABPHS}}
\label{sec:ROABP}
%Arithmetic Branching Program (ABP) is another 
%interesting model of computing polynomials. 
An ABP is a directed graph with $d+1$ layers of vertices
$\{V_0,V_1, \dots, V_{d}\}$ such that the edges are 
only going from $V_{i-1}$ to $V_i$ for any $i \in [d]$. 
As a convention, $V_0$ and $V_d$ have only one node each,
 let the nodes be $v_0$ and $v_d$ respectively.
A width-$w$ ABP has $\abs{V_i} \leq w$ for all $i \in [d]$.
Let the set of nodes in $V_i$ be $\{v_{i,j} \mid j \in [w]\}$.
All the edges in the graph have weights from $\F[\o{x}]$,
for some field $\F$. For an edge $e$, let us denote
its weight by $w(e)$.
For a path $p$ from $v_0$ to $v_d$,
its weight $w(p)$ is defined to be the product of weights of all the edges
in it, i.e.\ $\prod_{e \in p}w(e)$. 
Consider the polynomial $C(\o{x}) = \sum_{p \in \paths(v_0,v_d)} w(p)$ 
which is the sum of the weights of all the paths from $v_0$ to $v_d$.
This polynomial $C(\o{x})$ is said to be computed by the ABP.

It is easy to see that this polynomial is the same as
$D_0^T (\prod_{i=1}^{d-2} D_i ) D_{d-1} $,
where $D_0,D_{d-1} \in (\F[\o{x}])^w$ 
and $D_i$ for $1 \leq i \leq d-2$
is a $w \times w$ matrix such that 
\begin{eqnarray*}
D_0(\ell) &=& w(v_{0},v_{1,\ell}) \text{ for } 1 \leq \ell \leq w\\
D_i(k, \ell) &=& w(v_{i,k},v_{i+1,\ell}) \text{ for } 1 \leq \ell,k \leq w \text{ and } 1 \leq i \leq d-2\\
D_{d-1}(k) &=& w(v_{d-1,k},v_{d}) \text{ for } 1 \leq k \leq w
\end{eqnarray*}

An ABP is called a read once ABP (ROABP)
if the edge weights in the different layers 
are polynomials in disjoint sets of variables
(it is actually called oblivious ROABP, but we drop the word 
oblivious from now on). 
More formally, 
there exists an unknown partition of the variable set $\o{x}$
into $d$ sets $\{\lis{\o{x}}{,}{d}\}$
such that
in the corresponding matrix product 
$D_0 (\prod_{i=1}^{d-2} D_i ) D_{d-1} $,
the entries in $D_{i-1}$ are polynomials in variables ${\o{x}}_{i}$,
for all $i \in [d]$.
It is read once in the sense that in the corresponding 
ABP, any particular variable contributes to at most one edge 
on any path.

We work with the matrix representation of ABP. 
We will show a hitting-set for an ROABP 
$D_0 (\prod_{i=1}^{d} D_i ) D_{d+1} $ with 
$D_i$ being {\em invertible} matrices for all $i \in [d]$ and
all the matrices being {\em sparse} polynomials.
Hence, we name this model \emph{sparse-invertible-factor ROABP}.  
Like in the previous sections, we find a hitting-set by showing
a low-support concentration. 

For a polynomial $D$, let its sparsity  
$\sp(D)$ be the number of monomials in $D$
with nonzero coefficients and
let $\mu(D)$ be the maximum support 
of any monomial in $D$.
%, i.e.\ $\mu(D) := \displaystyle\max_{\o{x}^e \in \Sp(D)} \supp(e)$.

\begin{reptheorem}{thm:ROABPHS}
\label{thm:ROABPHS-rep}
Let $\o{x} = \o{x}_0 \sqcup \dotsm \sqcup \o{x}_{d+1}$, with $\abs{\o{x}} = n$.
Let $C(\o{x}) = D_0^T D D_{d+1} \in \F[\o{x}]$ be a polynomial
with $D(\o{x}) = \prod_{i=1}^{d} D_i(\o{x}_i)$, where
$D_0 \in \F^w[\o{x}_0]$ and $D_{d+1} \in \F^w[\o{x}_{d+1}]$
and for all $i \in [d]$, 
$D_i \in \F^{w \times w}[\o{x}_i]$ is an invertible matrix.
For all $i \in \{0, 1, \dots, d+1 \}$, $D_i$ has degree bounded
by $\delta$, $\sp(D_i) \leq s$ and $\mu(D_i) \leq \mu$.
Let $\ell := 1 + 2 \min\{ \ceil{\log (w^2 \cdot s)}, \mu \}$.
%and $M :=  n \log \delta(d s^{2w} + (n \delta)^{2 \ell})$.
Then there is a hitting-set of size 
%$O(M^2 \log M n^{\ell (w^2+2)})$ for $g(\o{x})$.
$\poly((n \delta s)^{\ell w^2})$ for $C(\o{x})$.
\emph{(For the proof, see Appendix~\ref{app:ROABP}.)}
\end{reptheorem}

\begin{remark}
If $\mu=1$, i.e.\ each $D_i$ is univariate or linear, then
we get poly-time for constant $w$.
\end{remark}
\noindent {\bf Proof Idea-} Here again, we find a hitting-set 
by proving low-support concentration, but with a different approach.
As all the matrices in the matrix product $D(\o{x}) = \prod_{i=1}^{d} D_i(\o{x}_i)$
are over disjoint sets of variables,
any coefficient in the polynomial $D(\o{x})$ can be uniquely written as a product of 
$d$ factors, each coming from one $D_i$. 
We start with the assumption that the constant term of each polynomial $D_i$, 
denoted by $D_{i {\bf 0}}$, is an invertible matrix.
Using this we define a notion of {\em parent} and {\em child}
between all the coefficients: If a coefficient can be obtained
from another coefficient by replacing one of its constant factors $D_{i {\bf 0}}$
with another term (with non-trivial support) from $D_i$, 
then former is called a parent
of the latter. 
Observe that if we want to do this replacement by a multiplication of some matrix, 
then $D_{i {\bf 0}}$ should be invertible.
Moreover, all the factors on its right side (or its left side) 
also need to be
% invertible , 
%hence they should also be 
constant terms in their respective matrices (this is because of non-commutativity).
For a coefficient, the set of matrices $D_i$ which contribute a non-trivial
factor to it, is said to form the {\em block-support} of the coefficient.

Our next step is to show that if a coefficient linearly depends on its descendants 
then the dependence can be lifted to its parent (by dividing and multiplying appropriate factors) i.e.\ its parent also linearly depends 
on its descendants. 
As the dimension of the matrix algebra is constant, 
if we take an appropriately large (constant) child-parent chain,
there will be a linear dependence among the coefficients in the chain.
As the dependencies lift to the parent, they can be lifted all the way up. 
By an inductive argument it follows that every coefficient depends on 
the coefficients with low-block-support. 
Now, this can be translated to low-support concentration in $D$,
if a low-support concentration is assumed in each $D_i$. 

To achieve low-support concentration in each $D_i$,
we use an appropriate shift. The sparsity of $D_i$ is used crucially 
in this step. 
To make $D_{i {\bf 0}}$ invertible, again an appropriate shift is used. 
Note that $D_{i {\bf 0}}$ can be made invertible by a shift only when 
$D_i$ itself is invertible, hence the invertible-factor assumption.

%\vspace{-0.25cm}
\section{Discussion}
%Proving low-support concentration for general 
%depth-$3$ multilinear circuits remains an open question. 
%As there are known lower bounds for this class of circuits \cite{RY09},
%we can hope that getting a small hitting-set should also
%be possible.
We conjecture that for depth-$3$ multilinear circuits,
low-support concentration can be achieved by an efficient shift.
A first question here is to remove the knowledge of the base sets in 
Theorem~\ref{thm:baseSetsHS}.

In the case of constant width ROABP, 
we could show constant-support concentration, but
only after assuming that the factor matrices are invertible. 
It seems that the invertibility assumption restricts the computing power of ABP
significantly.
It is desirable to have low-support concentration without the
assumption of invertibility. 
%Also, the dependence on the dimension $w$ is exponential, one
%should be able to reduce that. 

\section{Acknowledgements}
We thank Chandan Saha for suggestions to improve this paper. Several useful ideas about $\d$-distance circuits and base sets came up during discussions with him.

\bibliographystyle{alpha}
\bibliography{deltaDistance}

\begin{thebibliography}{GKKS13}

\bibitem[Agr05]{Agr05}
Manindra Agrawal.
\newblock Proving lower bounds via pseudo-random generators.
\newblock In {\em FSTTCS}, volume 3821 of {\em Lecture Notes in Computer
  Science}, pages 92--105, 2005.

\bibitem[ASS13]{ASS13}
Manindra Agrawal, Chandan Saha, and Nitin Saxena.
\newblock Quasi-polynomial hitting-set for set-depth- formulas.
\newblock In {\em STOC}, pages 321--330, 2013.

\bibitem[BOC92]{BOC92}
Michael Ben-Or and Richard Cleve.
\newblock Computing algebraic formulas using a constant number of registers.
\newblock {\em SIAM J. Comput.}, 21(1):54--58, 1992.

\bibitem[FSS13]{FSS13}
Michael~A. Forbes, Ramprasad Saptharishi, and Amir Shpilka.
\newblock Pseudorandomness for multilinear read-once algebraic branching
  programs, in any order.
\newblock {\em CoRR}, abs/1309.5668, 2013.

\bibitem[GKKS13]{GKKS13}
Ankit Gupta, Pritish Kamath, Neeraj Kayal, and Ramprasad Saptharishi.
\newblock Arithmetic circuits: A chasm at depth three.
\newblock {\em FOCS}, 2013.

\bibitem[Kro82]{Kro1882}
Leopold Kronecker.
\newblock {\em Grundzuge einer arithmetischen Theorie der algebraischen
  Grossen}.
\newblock Berlin, G. Reimer, 1882.

\bibitem[Mul12]{mul12}
Ketan Mulmuley.
\newblock Geometric complexity theory {V}: Equivalence between blackbox
  derandomization of polynomial identity testing and derandomization of
  {N}oether's normalization lemma.
\newblock In {\em FOCS}, pages 629--638, 2012.

\bibitem[RS05]{RS05}
Ran Raz and Amir Shpilka.
\newblock Deterministic polynomial identity testing in non-commutative models.
\newblock {\em Computational Complexity}, 14(1):1--19, 2005.

\bibitem[RY09]{RY09}
Ran Raz and Amir Yehudayoff.
\newblock Lower bounds and separations for constant depth multilinear circuits.
\newblock {\em Computational Complexity}, 18(2):171--207, 2009.

\bibitem[Sax09]{Sax09}
Nitin Saxena.
\newblock Progress on polynomial identity testing.
\newblock {\em Bulletin of the EATCS}, 99:49--79, 2009.

\bibitem[Sch80]{Sch80}
Jacob~T. Schwartz.
\newblock Fast probabilistic algorithms for verification of polynomial
  identities.
\newblock {\em J. ACM}, 27(4):701--717, October 1980.

\bibitem[SSS09]{SSS09}
Chandan Saha, Ramprasad Saptharishi, and Nitin Saxena.
\newblock The power of depth 2 circuits over algebras.
\newblock In {\em FSTTCS}, pages 371--382, 2009.

\bibitem[SY10]{SY10}
Amir Shpilka and Amir Yehudayoff.
\newblock Arithmetic circuits: A survey of recent results and open questions.
\newblock {\em Foundations and Trends in Theoretical Computer Science},
  5(3-4):207--388, 2010.

\end{thebibliography}

%%%%%%%%%%%%%%%%%%%%%%%%%%%%%%%%%%%%%%%%%%%%%%%%%%%%%%%%%%%%%%%%%%%%%%%%%%%
%Appendix

%%%%%%%%%%%%%%%%%%%%%%%%%%%%%%%%%%%%%%%%%%%%%%%%%%%%%%%%%%%%%%%%%%%%%%%%
\appendix

\section{Complete proofs of Section \ref{sec:dDist}}
\label{app:dDist}

\begin{replemma}{lem:hsFromlConc}
If $D(\overline{x})$ is $\ell$-concentrated, then there is a $n^{O(\ell)}$-time hitting-set for $C(\overline{x})$.
\end{replemma}
\begin{proof}%[Proof of Lemma~\ref{lem:hsFromlConc}]
\begin{claim}\label{c}
The \p% over the base field $ \F,
 $C(\overline{x}) = 0$ iff all of its $ (<\ell) $-support monomials have coefficient $ =0 $.
\end{claim}
This claim gives an $n^{O(\ell)}$-time hitting-set all $ C(\overline{x}) $: We can isolate each of the %$ \choose{[n]}{(<\ell)} $
$ (<\ell) $-support monomials and check whether their coefficient is $ 0 $. How do we isolate the $ (<\ell) $-support monomials? For each subset $S$ of cardinality $< \ell$, set the variables in the set $S$ to $1$ and the other variables to $0$. This isolates $\sum_{T \subseteq S}\coeff_C(x_T) = \coeff_C(x_S) + \sum_{T \subsetneq S}\coeff_C(x_T) $. If it was already known that the coefficient of the smaller subsets is $0$, then $\sum_{T \subseteq S}\coeff_C(x_T) = 0 \implies \coeff_C(x_S) = 0$. Thus, starting with coefficients of cardinality $0$, we test the coefficient of each set for nonzeroness. The hitting-set thus consists of substituting $\o{x} = (\lis{b}{,}{n}), b_i \in \{0, 1\}$ with $< \ell$ many $1$s. There are $\sum_{i=0}^{\ell -1} {n \choose i} = n^{O(\ell)}$ such tuples.

\begin{proof}[Proof of Claim \ref{c}]
The forward implication holds trivially. To prove the reverse implication, we express the coefficients of the polynomial $ C(\overline{x}) $ as a linear combination of its low-support coefficients.
\begin{eqnarray*}
C(\overline{x}) &=& \overline{a}^T \cdot \Dx\\
				&=& \overline{a}^T \cdot \left(\sum_{S\subseteq [n]} u_S x_S\right), \mbox{ where } u_S := \coeff_D(x_S)\\
				&=& \sum_{S \subseteq [n]} \overline{a}^T  u_S x_S\\
\end{eqnarray*}

Since $ \Dx $ is $ \ell $-concentrated, the coefficients of each of its monomials can be expressed as a linear combination of coefficients of low-support monomials, i.e.\ $ \forall S \subseteq [n], u_S = \sum_{R \subseteq [n], \abs{R} < \ell} \alpha_{R,S} u_R$. Hence,
\[
C(\overline{x}) = \sum_{S \subseteq [n]} \left(\sum_{R \subseteq [n], \abs{R} < \ell} \alpha_{R,S}\left(\overline{a}^T  u_R\right) \right) x_S.
\]

Let $ {\beta}_R := \coeff_C(x_R) = \o{a}^T  u_R $. Hence,
%Now, the coefficient $ C_R $ of the monomial $ x_R $ in the \p $C(\overline{x})$ is $\overline{a}^T  D_R$. Hence,

\[
C(\o{x}) = \sum_{S \subseteq [n]} \left(\sum_{R \subseteq [n], \abs{R} < \ell} \alpha_{R,S} {\beta}_R \right) x_S.
\]

This proves the claim.

\end{proof} 
\end{proof}

\begin{lemma}[Efficient Kronecker map \cite{Kro1882,Agr05}]%2
\label{lem:kronecker}
Consider a set $A$ of $n$-variate monomials with maximum individual
degree $d$ and $\abs{A} = a$.  
There exists a set of $N$-many monomial maps 
$\phi \colon \o{t} \to \{t^i\}_{i=1}^{N \log N}$, such that at 
least one of them separates the monomials in $A$, where $N := nda^2 \log(d+1)$.

%i.e.\ $\forall \o{t}^{e_1}, \o{t}^{e_2} \in A$, $\phi(\o{t}^{e_1}), \phi(\o{t}^{e_2})$. 
\end{lemma}
\begin{proof}
Since we want to separate the $n$-variate monomials with maximum individual degree $d$, we use the naive Kronecker map $\phi' \colon t_i \mapsto t^{(d +1)^{i-1}}$ for all $i \in [n]$. It can easily be seen that there is a 1-1 correspondence between the original monomials and the substituted monomials. But the degree of the new monomials can be very high.

Hence, we take the substitutions $\bmod(t^p-1)$ for many small primes $p$. 
In other words, the degrees are taken modulo $p$. Each prime $p$ leads to a different substitution of the variables $t_i$s. That is our set of candidate maps. We need to bound the number $N$ of primes that ensure that at least one substitution separates the monomials in $A$. We choose the smallest $N$ primes, say $\mathcal{P}$ is the set. 
By Prime Number Theorem, the highest value in the set $\mathcal{P}$ is $N \log N$.

To bound the number $N$ of primes: We want a $p$ in the set $\mathcal{P}$ such that 
$\forall i \ne j, \; d_i - d_j \not\equiv 0 \pmod{p}$, where $d_i$s are the degrees of $t$ in the $a$ monomials after the substitution $\phi'$. I.e.\, we want a $p$ such that $p \nmid \prod_{i \ne j} \left(d_i - d_j\right)$. 

There are $N$ such primes. Hence we want that $\prod_{p \in \mathcal{P}} p \nmid \prod_{i \ne j} \left(d_i - d_j\right)$. This can be ensured by setting $\prod_{p \in \mathcal{P}} p > \prod_{i \ne j} \left(d_i - d_j\right)$.
There are $(< a^2)$ such monomial pairs and each $d_i < (d+1)^{nd}$. Also, $\prod_{p \in \mathcal{P}} p > 2^N$. Hence, $N = nda^2 \log(d+1)$ suffices.
\end{proof}

\begin{replemma}{lem:finalHS}
If for a polynomial $ D(\o{x}) \in \H_k(\F)[\o{x}]$, there exists a set of $f(n)$-many maps from 
$\o{t}$ to $\{t^i\}_{i=1}^{g(n)}$ such that for at least one of the maps $\phi$, the shifted polynomial
$D(\o{x}+\phi(\o{t}))$ has $\ell$-concentration, then $C(\o{x}) = \o{a}^T D(\o{x})$, for any $\o{a} \in \F^k$, has an $n^{O(\ell)} f(n) g(n)$-time hitting-set.
\end{replemma}
\begin{proof}%[Proof of Lemma~\ref{lem:finalHS}]
We need to try $f(n)$-many maps, such that one of them will ensure
$\ell$-concentration in $D(\o{x}+\phi(\o{t}))$. 
Lemma~\ref{lem:hsFromlConc} shows that $\ell$-concentration implies 
a hitting-set of size $n^{O(\ell)}$ for $C(\o{x}+\phi(\o{t}))$. 
But each evaluation of $C(\o{x} + \phi(\o{t}))$ 
will be polynomial in $t$ with degree at most $n g(n)$.
This multiplies a factor of $n g(n)$ to the running time. 
So, total time complexity is $n^{O(\ell)} f(n) g(n)$.
\end{proof}

%In Lemmas \ref{lem:subcircuit} and \ref{lem:subcircuitl0}, $\monomials(\tower_j(X))$ denotes the set of all possible \ml monomials over the variable set $\variables(\tower_j(X))$.

\begin{replemma}{lem:subcircuit}
Let $\Fx \in \H_j(\F)[\o{x}]$ be a $\d$-distance circuit.
Let $F'(\o{x}) = F(\o{x} + \phi_{\d\ell_0}(\o{t}))$ be
its shifted version.
Let $V_{F'}$ be the subspace spanned by the phase-$(<j)$ coefficients of $\Fpx$.
Then, $\Fpx$ has $\ell_0$-concentration modulo $V_{F'}$.
I.e.\ $\forall $ monomial $T$ in $F'$,
$$ \coeff_{F'}(x_T)
\in \Span\{\coeff_{F'}(x_S) \mid S \subseteq \monomials(F'), \; \abs{S} < \ell_0\} 
+ V_{F'}, 
$$where, $\monomials(F')$ denotes the set of all possible \ml monomials that appear in $F'$.
\end{replemma}
\begin{proof}
The proof is by induction on the number of neighborhoods in
$F'$, $\norm{\nbd_j(F')}$.
We prove the theorem for intermediate subcircuits $E'$ of $F'$.
These subcircuits $E'$ have progressively more neighborhoods in the partition $\P_j$.

\emph{Base Case:} Any shifted $\d$-distance circuit,
$E'$ over $j$ coordinates, with $\norm{\nbd_j(E')} \le \ell_0$, 
has $\ell_0$-concentration modulo $V_{E'}$,
from Lemma~\ref{lem:subcircuitl0}.

\emph{Induction Hypothesis:} Any shifted $\d$-distance circuit, $E'$ over $j$ coordinates,
such that $\norm{\nbd_j(E')} \leq r$ ($r \geq \ell_0$), 
has $\ell_0$-concentration modulo $V_{E'}$.

\emph{Induction Step:} Let
$\norm{\nbd_j(E')} = r+1$.
% Let $V'$ denote the phase-$(<j)$ coefficients
%in $E'$. 
Let $E'_0$ be the subcircuit corresponding to the tower over
\emph{one} neighborhood in $\nbd_j(E')$.
Let $E'_1$
be the subcircuit corresponding to the tower over
the remaining neighborhoods in $\nbd_j(E')$, i.e.\ $\norm{\nbd_j(E'_1)} = r$.
Hence, $E' = E'_1 \star E'_0$.

We will now prove \lzc of an arbitrary monomial $T$ in $E'$.
Since $E'_0$ and $E'_1$ are disjoint towers, the variables in their supports are disjoint.
Let $T_0 = \variables(E'_0) \cap T$ and $T_1 = \variables(E'_1) \cap T$.
Hence, $T = T_0 \cup T_1$.
By the induction hypothesis,
$E'_1$ is $\ell_0$-concentrated
modulo $V_{E'_1}$.
Hence, 
$$\coeff_{E'_1}(x_{T_1}) \in  \Span\{\coeff_{E'_1}(x_S) \mid S \in \monomials(E'_1), \; \abs{S} < \ell_0\} + V_{E'_1},$$
where, $\monomials(E'_1)$ denotes the set of monomials in $E'_1$.
We know that $\coeff_{E'}(x_T) = \coeff_{E'_1}(x_{T_1}) \star \coeff_{E'_0}(x_{T_0})$
from Observation~\ref{obs:CoeffMulInTowers}.
Multiplying the above equation with $\coeff_{E'_0}(x_{T_0})$, we get,
$$\coeff_{E'}(x_{T}) \in  \Span\{\coeff_{E'_1}(x_S) \star \coeff_{E'_0}(x_{T_0}) \mid S \in \monomials(E'_1), \; \abs{S} < \ell_0\} + {V_{E'}}.$$
This holds because, by Observation \ref{obs:supSetPhase}, $\coeff_{E'_0}(x_{T_0}) \star V_{E'_1} \subseteq V_{E'}$.
% Hence, $\coeff_{E'_0}(x_{T_0}) \star V_{E'_1}  \subseteq V_{E'}$.
By Observation \ref{obs:CoeffMulInTowers},
the above equation is equivalent to
\begin{equation}
\label{eq:TinSpan}
\coeff_{E'}(x_T) \in  \Span\{\coeff_{E'}(x_{S \cup T_0}) \mid S \in \monomials(E'_1), \; \abs{S} < \ell_0\} + {V_{E'}}. \end{equation}
We will now prove that $\coeff_{E'}(x_{S \cup T_0})$ is \lzcd (modulo ${V_{E'}}$) for all monomials $S$ in $E'_1$. This will prove that $\coeff_{E'}(x_T)$ is \lzcd (modulo ${V_{E'}}$).

%Now for any $S \in \monomials(E'_1)$, 
%such that 
Let $\norm{\nbd_j(S)}$ be the number of neighborhoods in $\nbd_j(S)$.
Since $\abs{S} < \ell_0$, we know $\norm{\nbd_j(S)} < \ell_0$.
Hence, $\norm{\nbd_j(S \cup T_0)} \leq \ell_0$. 
Let the subcircuit corresponding to $\tower_j(S \cup T_0)$ be $E'_2$.
Let the subcircuit $E'_3$ be such that $E' = E'_2  \star E'_3$.
From Observation \ref{obs:CoeffMulInTowers},
%$\coeff_{E'}(x_{S \cup T_0})$ can be written 
%as a product of its coefficient coming from its tower with 
%a constant term from outside its tower, i.e.\
\begin{equation*}
\coeff_{E'}(x_{S \cup T_0}) = \coeff_{E'_2}(x_{S \cup T_0}) \star \coeff_{E'_3}(x_{\varnothing}).
%\label{eq:coeffbreak}
\end{equation*}
As $\norm{\nbd_j(E'_2)} \leq \ell_0$, from Lemma~\ref{lem:subcircuitl0}, we have $\ell_0$-concentration in $E'_2$ modulo $V_{E'_2}$.
%its phase-$(<j)$ coefficients. So, 
$$\coeff_{E'_2}(x_{S \cup T_0}) \in  \Span\{\coeff_{E'_2}(x_R) \mid R \in \monomials(E'_2), \; \abs{R} < \ell_0\} + {V_{E'_2}}.$$
%where $V_{E'_2}$ is the span of phase-$(<j)$ coefficients in $E'_2$.
Multiplying throughout with $\coeff_{E'_3}(x_{\varnothing})$,
%Using Equation~(\ref{eq:coeffbreak}) with this we get, 
$$\coeff_{E'}(x_{S \cup T_0}) \in  \Span\{\coeff_{E'_2}(x_R) \star \coeff_{E'_3}(x_{\varnothing}) \mid R \in \monomials(E'_2), \; \abs{R} < \ell_0\} + {V_{E'}}.$$
%Here, we are again using that
This is because $\coeff_{E'_3}(x_{\varnothing}) \star V_{E'_2}  \subseteq V_{E'}$.
The equation is equivalent to
\begin{equation}
\label{eq:ST0}
\coeff_{E'}(x_{S \cup T_0}) \in \Span\{\coeff_{E'}(x_R) \mid R \in \monomials(E'), \; \abs{R} < \ell_0\} + {V_{E'}}.
\end{equation}
Using Equations (\ref{eq:TinSpan}) and (\ref{eq:ST0}) we get that $\coeff_{E'}(T)$ is in the span of coefficients in $E'$ with support $< \ell_0$ modulo $V_{E'}$. 
Hence, we get $\ell_0$-concentration of $E'$ modulo $V_{E'}$.
\end{proof}

%%%%%%%%%%%%%%%%%%%%%%%%%%%%%%%%%%%%%%%%%%%%%%%%%%%%%%%%%
%%%%%%%%%%%%%%%%%%%%%%%%%%%%%%%%%%%%%%%%%%%%%%%%%%%%%%%%%%%%%%%%%%%%%%
%%%%%%%%%%%%%%%%%%%%%%%%%%%%%%%%%%%%%%%%%%%%%%%%%%%%%%%%%%%%%%%%%%%%%%%%%%%%

\section{Proof of Lemma \ref{lem:subcircuitl0} (Small-Neighborhood $\ell_0-concentration$)}
\label{app:subcircuitl0}

We will now show $\ell_0$-concentration in a $\d$-distance circuit $E'$
modulo its phase-$(<j)$ coefficients,
where the number of neighborhoods in the partition $\P_j$ is $\leq \ell_0$.
The basic idea of the proof is to take dependencies among the coefficients
of $E$ and lift them to get new dependencies of $E'$.

\begin{replemma}{lem:subcircuitl0}
Let $\Ex \in \H_j(\F)[\o{x}]$ be a $\d$-distance circuit
with $\norm{\nbd_j(E')} \leq \ell_0$.
Let $E'(\o{x}) = E(\o{x} + \phi_{\d\ell_0}(\o{t}))$ be
its shifted version.
Let $V_{E'}$ be the subspace spanned by the phase-$(<j)$ coefficients of $\Epx$.
Then, $\Epx$ has $\ell_0$-concentration modulo $V_{E'}$.
I.e.\ $\forall $ monomial $T$ in $E'$,
$$ \coeff_{E'}(x_T)
\in \Span\{\coeff_{E'}(x_S) \mid S \subseteq \monomials(E'), \; \abs{S} < \ell_0\} 
+ V_{E'}.
$$
%where, $\monomials(E')$ denotes the set of all possible \ml monomials over the variable set $\variables(E')$.
%Let $\Epx$ be a subcircuit of $\Dpx$ corresponding to a tower 
%over a set of neighborhoods in partition $\P_j$
%such that $\norm{\nbd_j(E')} \leq \ell_0$. 
%Let $V'$ be the subspace spanned by the coefficients of the monomials in phase-$(\le j-1)$ in $\Epx$.
%Then, $E'$ has $\ell_0$-concentration modulo $V'$, 
%i.e.\ $\forall S \in \monomials(E')$
%$$ \coeff_{E'}(x_S)
%\in \Span\{\coeff_{E'}(x_T) \mid T \subseteq \monomials(E'), \; \abs{T} < \ell_0\} 
%+ {V'}. 
\end{replemma}
\begin{proof}
Let the matrices $U$ and $U' \in \F^{[j] \times 2^{\n}}$ 
be such that their columns represent the coefficient 
vectors in $E$ and $E'$ respectively.
Let us recall the relation between $U$ and $U'$ from Section~\ref{subsec:generalApproach},
\begin{equation}
\label{eq:transfer}
U = U' \TM, 
\end{equation}
where $\TM$ is the transfer matrix defined in Section~\ref{subsec:generalApproach}. 
Let $\Null(U)$ denote the nullspace of $U$.
Consider a vector $\alpha \in \Null(U)$, i.e.\
$U \alpha = 0$.
%Using relation between $U$ and $U'$ we get,
Using Equation (\ref{eq:transfer}), $U' \TM \alpha = 0$.
Thus, $\TM \alpha \in \Null(U')$.
%Let $\beta := \TM\alpha$.

Let us now study the dependencies we want for $U'$ to show 
the appropriate concentration. 
Let $\S_1$ be the set of monomials which have support $\geq \ell_0$ and
are in phase-$j$. 
Let $\S_0$ be the set of all other monomials. 
Let $U'_0 \in \F^{[j] \times \S_0}$ and $U'_1 \in \F^{[j] \times \S_1}$
be the submatrices of $U'$ consisting of coefficient vectors
of monomials in $\S_0$ and $\S_1$ respectively.
To prove the lemma,
we need to show that the columns of $U'_1$ are in the span of the columns of $U'_0$.
Let $\beta \in \F(t)^{2^{\n}}$ be a dependency among the columns of $U'$, i.e.\
$U' \beta = \mathbf{0}$.
We can break this equation into two parts as: 
$U'_1 \beta_1 = - U'_0 \beta_0$,
where $\beta_0$ and $\beta_1$ are the subvectors of $\beta$,
indexed by the monomials in
$\S_0$ and $\S_1$ respectively.
In other words, 
\begin{equation}
\label{eq:nullv}
U'_1 \beta_1 \equiv \mathbf{0} \pmod{U'_0},
\end{equation}
where, by ``mod $U'_0$" we mean modulo column-span$(U'_0)$.
Now, we want to show that $U'_1 \equiv \mathbf{0} \pmod{U'_0}$.

Let $\abs{\S_1}$, the number of columns in $U'_1$ be $N^*$.
Suppose there exists a set of $N^*$ many
null-vectors of $U'$.
Let $B \in \F^{2^{\n} \times [N^*]}$, be the matrix with these null-vectors as its columns.
Let $B_1$ be the submatrix of $B$ obtained by picking the rows
indexed by the monomials in
$\S_1$. Using Equation~(\ref{eq:nullv}),
\[
U'_1 B_1 \equiv \mathbf{0} \pmod{U'_0}.
\]
If $B_1$ is an invertible matrix then,
\begin{eqnarray*} 
U'_1 &\equiv& \mathbf{0} B_1^{-1} \pmod{U'_0} \\ 	
\implies U'_1 &\equiv& \mathbf{0} \pmod{U'_0}, 
\end{eqnarray*} 
which is the desired result.  

Now, to show the existence of such a set of null-vectors $B$,
we will take the null-vectors of $U$ and lift them. 
Here, we will use the notation of a matrix, to also denote the set of 
its column-vectors. 
As seen earlier, if $\nul^* \subseteq \Null(U)$, 
then $B := \TM \nul^* \subseteq \Null(U')$ 
(by `$\subseteq$' we mean column-wise). 
Let $\TM_1$ be the submatrix of $\TM$, obtained by picking the rows of $\TM$
indexed by the monomials in $\S_1$.
%Then $B_1 := \TM_1 \nul^* \subseteq \Null(U'_1 \pmod{U'_0})$.
We intend to show that there exists a set of $N^*$ many null-vectors
$\nul^* \subseteq \Null(U)$ such that $B_1 := \TM_1 \nul^*$ is an invertible matrix.  

For this, let us analyze the monomials in $\S_1$.
Recall that any phase-$j$ monomial has at most $\delta$
support from any neighborhood in partition $\P_j$.
As $\norm{\nbd_j(E)} \leq \ell_0$, any phase-$j$ monomial has
degree at most $\delta \ell_0$. 
Hence, monomials in $\S_1$ have degree between $\ell_0$ and
$\delta \ell_0$.
%All the monomials between degree $\ell_0$ and $\delta \ell_0$ 
%need not be in $\S_1$ but they can be. 
It suffices to show $\TM_1 \nul^*$ is an invertible matrix,
assuming $\TM_1$ has rows indexed by
all the monomials with degree between $\ell_0$ and $\delta \ell_0$
(If this $\TM_1 \nul^*$ is invertible, then,
for every subset $\widehat{\TM}_1$ of rows of $\TM_1$,
there exists a subset $\widehat{\nul}^*$ of null-vectors ${\nul}^*$
such that $\widehat{\TM}_1 \widehat{\nul}^*$ is invertible).
%(The other cases of $\S_1$ will be subsumed in this).
Now, we redefine $N^*$ to be the number of such monomials.

Recall that $\TM_1(S,T) \neq 0$ only when $T \subseteq S$.
Since the rows of $\TM_1$ are indexed by the monomials of degree at most 
$\delta \ell_0$, its columns indexed by the monomials of degree $>\delta \ell_0$
will be all zero. 
Hence, the rows of $\nul^*$ indexed by the monomials of degree $> \delta \ell_0$ 
can be ignored.
We thus redefine $\TM_1 $ and $\nul^*$ as their truncated versions.

So, we can as well assume that $\nul^* \subseteq \Null(U_{\delta \ell_0})$,
where $U_{\delta \ell_0}$ is the submatrix of $U$,
obtained by picking the columns indexed by
the monomials of degree $\leq \delta \ell_0$. 

Let $N$ be the number of all $n$-variate monomials of degree $\leq \delta \ell_0$.
Let $\nul := \Null(U_{\delta \ell_0})$.

Now, in Lemma~\ref{lem:invertible} we show that
for any vector-space $\nul \subseteq \F^N$
of dimension $\ge (N-k)$,
there exists
$\nul^* \subseteq \nul$ such that $\TM_1 \nul^*$ is invertible.
As $U_{\delta \ell_0}$ has rank at most $j$,
$\nul$ has dimension at least $(N-j) \geq (N-k)$ as required in Lemma \ref{lem:invertible}.

This completes the proof.
%The crucial part in the proof is that while $\TM_1$ is a matrix over
%$\F(t)$, $\nul^*$ is a matrix over constants as $U$ itself is a matrix over constants. 
\end{proof}

To show that there exists such a set of vectors $\nul^*$, 
now we will look at some properties of the transfer matrix. 
This is a simplified and improved proof of Theorem 13 in \cite{ASS13}.

\begin{lemma}[Transfer matrix theorem]\label{lem:rowComb}%9
Consider a matrix $ M_{n, \ell} \; (\ell \le n)$, with rows indexed by all possible $ n $-variate multilinear monomials with support size $ \ge \ell$
and columns indexed by all possible $ n $-variate multilinear monomials. Let $M_{n, \ell}(S, T) = 1$ if $T \subseteq S$ and $0 $ otherwise. Then any vector formed by a nonzero linear combination of the rows over any field $\F$ has at least $2^{\ell}$ nonzero entries.
\end{lemma}
\begin{proof}
Proof is by induction on the number of variables, $n$.

\emph{Base case:} There is only one variable, $x_1$. Then $\ell$ is $0$ or $1$.

When $\ell = 1$, the matrix $M_{1, 1}$ has only one row, indexed by the monomial $x_1$. The matrix $M_{1, 1}$ has two columns. They are indexed by the empty set $\varnothing$ and the monomial $x_1$. So, $M_{1, 1} =
\left(\begin{smallmatrix}
1&1
\end{smallmatrix}\right)$; %$(1, 1)$.
which clearly satisfies the lemma statement.

When $ \ell = 0$, both the rows and the columns are indexed by the empty set $\varnothing$ and the monomial $x_1$. The matrix
$M_{1, 0} =
\left(\begin{smallmatrix}
1&0\\1&1
\end{smallmatrix}\right)$.
Since the two rows are linearly independent, any nonzero linear combination will have at least $1 = 2^0$ nonzero entry.

\emph{Induction Hypothesis:} Assume that for number of variables = $n-1$ and for all $\ell \le n-1$, any nonzero linear combination of the rows has at least $2^{\ell}$ nonzero entries.

\emph{Induction Step:} We have to prove the property for $M_{n, \ell }$ for all $\ell \le n$.

%When $\ell = n$, $M_{n, \ell }$ has a single row. The columns are indexed by all the subsets of $[n]$ and hence, the row has all $1$'s. Thus, $M_{n, n}$ satisfies the statement.
%
%In th rest of the proof, we can assume that $\ell \le n-1$.

The rows of the matrix $M_{n, \ell}$ are partitioned into two sets:
${\mathcal S}_1$, the set of rows whose indices do not contain $x_n$
and ${\mathcal S}_2$, the set of rows whose indices contain $x_n$.
The columns of the matrix $M_{n, \ell}$ are similarly partitioned into two sets:
${\mathcal T}_1$, the set of columns whose indices do not contain $x_n$
and ${\mathcal T}_2$, the set of columns whose indices contain $x_n$.
Then $\mnl$ is divided into four blocks: $\left\{\mnl\left({\mathcal S}_i, {\mathcal T}_j\right)\right\}, \left(i, j \in  \{ 1, 2 \} \right)$. Clearly,
\begin{eqnarray}
\mnlst{1}{1} &=& M_{n-1, \ell}\label{smallermatrix11}\\
\mnlst{1}{2} &=& \mathbf{0}\label{smallermatrix12} \mbox{ and}\\
\mnlst{2}{1} &=& \mnlst{2}{2}.\label{smallermatrix2122}
\end{eqnarray}
Note that, if $\ell = n$, Equation~(\ref{smallermatrix11}) still holds, since then, ${\S}_1 = \varnothing$.
Equation \ref{smallermatrix2122} holds because $T \cup \{n\} \subseteq S \cup \{n\}$ iff $T \subseteq S \cup \{n\}$, where $S, T \subseteq [n-1]$.

We break the linear combination: $\sum_{S \in \left({\mathcal S}_1 \cup {\mathcal S}_2\right)} c_S M_S = \sum_{S \in {\mathcal S}_1} c_S M_S + \sum_{S \in {\mathcal S}_2} d_S M_S$ where $c_S, d_S \in \F$ and $M_S$ is the row of $\mnl$ indexed by the set $S$.

\begin{eqnarray*}
\sum_{S \in \left({\mathcal S}_1 \cup {\mathcal S}_2\right)} c_S M_S
&=& \sum_{S \in {\S}_1} c_S \Mbreak{M(S, {\T}_1)}{M(S, {\T}_2)}
+ \sum_{S \in {\S}_2} d_S \Mbreak{M(S, {\T}_1)}{M(S, {\T}_2)}\\
&=& \sum_{S \in {\S}_1} c_S \Mbreak{M(S, {\T}_1)}{\mathbf 0}
+ \sum_{S \in {\S}_2} d_S \Mbreak{M(S, {\T}_1)}{M(S, {\T}_1)}\\
&=&\Mbreak{\displaystyle\sum_{S \in {\S}_1} c_S M(S, {\T}_1) + \sum_{S \in {\S}_2} d_S M(S, {\T}_1)}{\displaystyle\sum_{S \in {\S}_2} d_S M(S, {\T}_1)}\\
&=&\Mbreak{\o{C} + \o{D}}{\o{D}},
\end{eqnarray*}
where, $\o{C}:= \sum_{S \in {\S}_1} c_S\; M(S, {\T}_1)$
and $\o{D}:= \sum_{S \in {\S}_2} d_S\; M(S, {\T}_1)$ are row vectors in $\F^{2^{[n-1]}}$.
The second equality holds from Equations (\ref{smallermatrix12}) and (\ref{smallermatrix2122}).

Let us first consider the case when $\o{C}\ne \mathbf{0}$.
%Now, the contribution of the rows in $\S_1$ may be zero. I.e.\ $c_S = 0, \forall S \in {\S}_1$.
Since $\o{C}$ is a nonzero linear combination of rows in $\mnlst{1}{1}$,
from Equation (\ref{smallermatrix11})
and by the induction hypothesis,
it has $\ge 2^{\ell}$ nonzero entries.
For any index $T \in {\T}_1$, if $\o{C}(T) \ne 0$, but $(\o{C} + \o{D})(T) = 0$, then $\o{D}(T) \ne 0$.
Hence, $\Mbreak{\o{C} + \o{D}}{\o{D}}$ has at least as many nonzero entries as $\o{C}$, i.e.\ $\geq 2^{\ell}$.
% Hence, in $\mnlst{1}{1}$, if a nonzero entry of the linear combination becomes $0$ after the addition of the linear combination in $\mnlst{2}{1}$, then, this loss is exactly compensated by $\mnlst{2}{2}$.

But, we cannot use the induction hypothesis if $\o{C} = \mathbf{0}$,
i.e.\ if ${\S}_1 = \varnothing $ or $c_S = 0, \forall S \in {\S}_1$.
In this case, we have to show that there are $\ge 2^{\ell}$ nonzero entries in
$\Mbreak{\o{D}}{\o{D}}$,
i.e.\ that there are $\ge 2^{\ell-1}$ nonzero entries in $\o{D}$
(assuming $\ell\ge 1$).
For this, observe that $\mnlst{2}{1} = M_{n-1, \ell - 1}$.
Hence, we can use the induction hypothesis.

If $\ell = 0$, we need to show that
that there are $\ge 1$ nonzero entries in $\o{D}$.
For this, observe that ${M_{n, 0}\left({\S}_2, {\T}_1\right)} = M_{n-1, 0}$.
Hence, we can use the induction hypothesis.
\end{proof}

Now, from this property of the transfer matrix, 
we will conclude linear independence of rows of a truncated
transfer matrix.

\begin{lemma}
Consider the matrix $M_{n, \ell }$ described in 
Lemma~\ref{lem:rowComb}.
Let us mark any set of at most $2^{\ell}-1$ columns in $M_{n, \ell }$. 
Let $M'_{n, \ell }$ denote the submatrix of $M_{n, \ell }$ consisting of 
all the unmarked columns. 
The rows of $M'_{n, \ell }$ are linearly independent.
\label{lem:mark}
\end{lemma}
\begin{proof}
Lemma~\ref{lem:rowComb} shows that any nonzero linear combination 
of the rows of $M_{n, \ell }$ has at least $2^{\ell}$ nonzero entries. 
$M'_{n, \ell }$ has at most $2^{\ell}-1$ columns missing from $M_{n, \ell }$.
So, any nonzero linear combination 
of the rows of $M'_{n, \ell }$ has at least one nonzero entry. 
In other words, the rows of matrix $M'_{n, \ell }$ are linearly independent. 
\end{proof}

Now, we are ready to prove the invertibility requirement in Lemma \ref{lem:subcircuitl0}. 
Recall that 
%$\ell_0 := \log (k +1)$.
$N$ denotes the number of all $n$-variate 
multilinear monomials with support size $\le \delta \ell_0 $ 
and $N^*$
denotes the number of all $n$-variate multilinear monomials with support size 
between $\ell_0$ and $\delta \ell_0 $. 
The truncated transfer matrix (introduced in the proof of Lemma \ref{lem:subcircuitl0}) $\TM_1$
has dimension $N^* \times N$. 
Now, we show that the truncated transfer matrix $\TM_1$ 
multiplied by 
an appropriate set of $N^*$-many null-vectors gives an invertible matrix.

\begin{lemma}[Transfer matrix action]
\label{lem:invertible}
Consider the truncated transfer matrix $\TM_1$. 
Given any space $\nul \subseteq \F^N$ of dimension at least $N-k$, 
there exists a set of $N^*$ vectors in it, denoted by
$\nul^*_{N \times N^*}$, 
such that $\TM_1 \nul^*$ is an invertible matrix. 
\end{lemma}
\begin{proof}%[Proof of Lemma~\ref{lem:invertible}]
From Section~\ref{subsec:generalApproach}, recall
that $\TM_1$ can be written as the product
$A' M_1 A$. 
Here $A'_{N^* \times N^*}$ and $A_{N \times N}$
are diagonal matrices with $A(T,T) = (-1)^{\abs{T}} \frac{1}{\phi_{\delta \ell_0}(t_T)}$.
Recall that $\phi_{\delta \ell_0}$ is a univariate map, 
which sends all monomials $t_T$ to distinct powers of $t$, when $\abs{T} \leq \delta \ell_0$.
Hence, $\phi_{\delta \ell_0}$ gives a total ordering on the monomials $t_T$. 
From now on, we assume that the columns of matrix $M_1$
are arranged in an increasing order according to the ordering given by $\phi_{\delta \ell_0}$.

$M_1$ is a matrix with columns indexed by monomials with support size
$\le \delta \ell_0$ and rows indexed by monomials with support size
between $\ell_0$ and $\delta \ell_0$.
Also, 
$$ 
M_1(S,T) = \begin{cases}
		1 & \text{if } T \subseteq S,\\
		0 & \text{otherwise} 
		\end{cases}
$$

Take a basis of the space $\nul$ of size $N-k$ and with an abuse of notation,
denote it by $\nul$.
In the matrix form $\nul$ has dimension $N \times (N-k)$.
The rows of matrix $\nul$ are also arranged in an increasing order
according to the ordering given by $\phi_{\delta \ell_0}$.

Any linear combination of the basis vectors 
remains in the same space. Hence, we can do column operations in the
matrix $\nul$. 
We can assume that after the column operations $\nul$ has a 
lower triangular form.
To be more precise, we can assume that
in any column of $\nul$ the first nonzero entry is $1$.
And if $i_j$ denotes the index of the first nonzero entry in 
$j$-{th} column, then $i_1 < i_2 < \dotsm < i_{N-k}$.
Clearly, the rows of $\nul$, given by the indices $I := \{i_1,i_2, \dots, i_{N-k}\}$,
are independent. 
The other rows corresponding to the indices $I' := [N] - I$
 are dependent on the $I$ rows.
Mark the columns in $M_1$ corresponding 
to the indices in $I'$. 
We know $\abs{I'} = k \leq 2^{\ell_0} -1$.
Now, we apply Lemma~\ref{lem:mark} to the matrix $M_1$.
Note that in Lemma~\ref{lem:mark}, the matrix $M_{n, \ell }$ has 
columns corresponding to all the multilinear monomials while
our matrix $M_1$ has columns only corresponding to monomials of
support size $\leq \delta \ell_0$.
So, we cannot directly apply Lemma~\ref{lem:mark} to $M_1$.
However, note that $M_1$ has rows corresponding
to monomials of support size between $\ell_0$ and $\delta \ell_0$. 
Hence, any column with monomial support size $> \delta \ell_0$ will
be a zero column. 
So, we can ignore the zero columns, and Lemma~\ref{lem:mark}
%holds true for $M_1$. 
%We conclude that 
implies: the rows of $M'_1$ are linearly independent,
where $M'_1$ denotes the matrix formed by 
the unmarked columns of $M_1$.
In other words, there exists a set of $N^*$ unmarked 
columns in $M_1$, which are linearly independent. 
Let the set of indices corresponding to these columns be $I^*$.

Recall that $I^* \subseteq I$.
So, we can choose a set of $N^*$ columns from $\nul$
such that the set of their first nonzero indices is $I^*$.
Let the matrix corresponding to these columns be $\nul^*_{N \times N^*}$.
Now, we consider the square matrix given by $R_{N^* \times N^*} := M_1 A \nul^*$.
We claim that $\det(R) \neq 0$.
To prove the claim we look at the lowest 
degree term in $\det(R)$.
Let $R_j$ be the $j$-{th} column of $R$ for $1 \leq j \leq N^*$.
Viewing $R_j$ as a polynomial in $\F^{N^*}[t]$, let
$R_{j0}$ be the coefficient of lowest degree term in $R_j$. 
Let us define a matrix $R_0$ whose $j$-{th} column is $R_{j0}$
for $1 \leq j \leq N^*$.
Clearly, $\det(R_0)$ is the coefficient of the lowest degree term in $\det(R)$, 
in the case when $\det(R_0) \neq 0$.
We claim that $\det(R_0) \neq 0$.

We know that $R_{j} = M_1 A \nul^*_j$, 
where $\nul^*_j$ is the $j$-{th} column of $\nul^*$.
Let $i_j$ be the index of the first nonzero entry in $\nul^*_j$.
As the entries in $A \nul^*_j$ have powers of $t$ in an strictly increasing order, 
the coefficient of the least term in $M_1 A \nul^*_j$
is clearly $\pm M_{i_j}$, where $M_{i_j}$ is the ${i_j}$-{th} column of $M_1$.
So, $R_{j0} = \pm M_{i_j}$ and hence $R_0 = M_{I^*}$ (upto a multiplicative factor of $(-1)$ to the columns), 
where $M_{I^*}$ is the submatrix of $M_1$
corresponding to $I^*$ columns.
As these columns are linearly independent, we get that
$\det(R_0) \neq 0$.
Which in turn implies that $\det(R) \neq 0$.
Hence, the product matrix $A' M_1 A \nul^* = \TM_1 \nul^*$
is invertible. 
\end{proof}

%%%%%%%%%%%%%%%%%%%%%%%%%%%%%%%%%%%%%%%%%%%%%%%%%%%%%%%%%
%%%%%%%%%%%%%%%%%%%%%%%%%%%%%%%%%%%%%%%%%%%%%%%%%%%%%%%%%%%%%%%%%%%%%%
%%%%%%%%%%%%%%%%%%%%%%%%%%%%%%%%%%%%%%%%%%%%%%%%%%%%%%%%%%%%%%%%%%%%%%%%%%%%

\section{Complete proofs of Section \ref{sec:baseSets}}
\label{app:baseSets}

%First, we will prove the following lemma
%which says that a linear dependence among some polynomials in
%$\H_k(\F)[\o{x}]$ implies a dependence among their coefficients.

\begin{lemma}[Circuits to coefficients]
\label{lem:polyToCoeffs}
Let $D_0(\o{x}), D_1(\o{x}), \dots, D_h(\o{x})$
be multilinear polynomials in $\H_k(\F)[\o{x}]$ for some field $\F$, 
where $\o{x}=\{\lis{x}{,}{n}\}$.
Let $D_i(\o{x}) = \sum_{S \subseteq [n]} u_{iS} x_S$, $\forall i \in \{0,1, \dots, h\}$,
where $u_{iS} \in \H_k(\F)$. If
$D_0(\o{x}) \in \Span_{\F(\o{x})}\{D_i(\o{x}) \mid i \in [h]\}$
then
$$\{u_{0S} \mid S \subseteq [n]\} \in \Span_{\F} \{u_{iS} \mid i \in [h], \; S \subseteq [n]\}.$$
\end{lemma}
\begin{proof}
Let us define a field $\F' := \F(\o{x})$.
$D_0 \in \Span_{\F'}\{D_i \mid i \in [h]\}$
implies that any null-vector for the vectors $D_1, D_2, \dots,D_h$ 
is also a null-vector for $D_0$, i.e.\
$$ \{ \alpha \in \H_k(\F') \mid \alpha^T \cdot D_i = 0, \; \forall i \in [h] \}
 \subseteq
 \{ \alpha \in \H_k(\F') \mid \alpha^T \cdot D_0 =0\}.
 $$
So, the statement is also true when the
 vector $\alpha$ coming from $\H_k(\F)$, 
i.e.\ 
\begin{equation}
 \{ \alpha \in \H_k(\F) \mid \alpha^T \cdot D_i = 0, \; \forall i \in [h] \}
 \subseteq
 \{ \alpha \in \H_k(\F) \mid \alpha^T \cdot D_0 =0\}.
 \label{eq:nullVectors}
 \end{equation}
It is easy to see that 
the set of null-vectors for a vector $D_i \in \H_k(\F')$,
 which are coming from $\H_k(\F)$, 
is the same as the intersection of the sets of null-vectors
of the coefficient vectors in $D_i$, i.e.\
\begin{equation} \{ \alpha \in \H_k(\F) \mid \alpha^T \cdot D_i =0\} = 
\{\alpha \in \H_k(\F) \mid \alpha^T \cdot u_{iS} = 0, \; \forall S \in [n] \}.
\label{eq:nullForCoeffs}
\end{equation}
Using Equations~(\ref{eq:nullVectors}) and (\ref{eq:nullForCoeffs}),
we can write,
\begin{equation*}
 \{ \alpha \in \H_k(\F) \mid \alpha^T \cdot u_{iS} = 0, \;  \forall i \in  [h], \; \forall S \subseteq [n]\} 
 \subseteq
 \{ \alpha \in \H_k(\F) \mid \alpha^T \cdot u_{0S} =0, \; \forall S \subseteq [n] \}.
\end{equation*}
Hence, we can write, for any $T \subseteq [n]$,
\begin{equation*}
 \{ \alpha \in \H_k(\F) \mid \alpha^T \cdot u_{iS} = 0, \; \forall i \in [h], \; \forall S \subseteq [n]\} 
 \subseteq
 \{ \alpha \in \H_k(\F) \mid \alpha^T \cdot u_{0T} =0 \}.
\end{equation*}
This clearly implies, by linear algebra, that for any $T \subseteq [n]$,
%\TODO{should we put some argument?}
\begin{equation*}
u_{0T} \in \Span_{\F} \{ u_{iS} \mid i \in [h], \; S \subseteq [n]  \}.
\end{equation*}
\end{proof}

%\begin{notation}
%Now, let $\Dx \in \H_k(\F)[\o{x}]$ be a polynomial having
%$m$-base-sets-$\delta$-distance.
%Let the base sets for $\Dx$ be $\{\lis{\o{x}}{,}{m}\}$.
%We will denote the set of indices corresponding to the variables
%in $\o{x}_i$ by $I_i$.
%Let us look at $\Dx$ as a polynomial over variable set $\o{x}_m$
%i.e.\ $\Dx = \sum_{S \in I_m} D_S x_S$, 
%where $D_S \in \H_k(\F)[\o{x}_1, \o{x}_2, \dots, \o{x}_{m-1}]$.
%\end{notation}
%One crucial property that will be used is that
%$m$-base-sets-$\delta$-distance property holds {\em recursively},
%i.e.\ if $D$ has $m$-base-sets-$\delta$-distance then
%$D_S$ has $(m-1)$-base-sets-$\delta$-distance.
%We will prove this in the following lemma. 

\begin{lemma}
\label{lem:recursive}
Let $\Dx$ has $m$-base-sets-$\delta$-distance with base sets
$\{\lis{\o{x}}{,}{m}\}$. Let us write 
$\Dx = \sum_{S \subseteq I_m} u_S x_S$, where $I_m$ denotes 
the set of indices corresponding to the variable set $\o{x}_m$ and
for all $S \subseteq I_m$, 
$u_S \in \H_k(\F)[\o{x}_1, \o{x}_2, \dots, \o{x}_{m-1}]$.
Then for any $S \subseteq I_m$,
$u_S$ has $(m-1)$-base-sets-$\delta$-distance
with base sets $\{\lis{\o{x}}{,}{m-1}\}$.
\end{lemma}
\begin{proof}
It is easy to see that $u_S$ is an evaluation of 
a derivative of $D$. To be more precise, 
$u_S = \frac{\partial D}{\partial x_S} {|}_{(x_j =0, \; \forall j \in I_m - S)} $,
where $\frac{\partial}{\partial x_S} := \circ_{j \in S} \frac{\partial}{\partial x_j}$.

Now, $D$ is a multilinear polynomial computed by a
$\Pi\Sigma$ circuit over $\H_k(\F)$, hence for each product gate,
a variable occurs in at most one of its input wires. 
If we take the derivative with respect to that variable, the corresponding wire 
vanishes from the circuit. So, the circuit for 
$u_S$ is essentially the same as that of $D$ except
some input wires to the product gates are missing
and also some variables are replaced with $0$.
Hence, $u_S$ when restricted to any variable set $\o{x}_i$,
still has $\delta$ distance, $\forall i \in [m-1]$.
Moreover, the variables from $\o{x}_m$ are not present in $u_S$.
Hence, $u_S$ has $(m-1)$-base-sets-$\delta$-distance
with base sets $\{\lis{\o{x}}{,}{m-1}\}$.
\end{proof}

\begin{lemma}
\label{lem:baseSetConc}
Let a polynomial $\Dx \in \H_k(\F)[\o{x}]$ have 
$m$-base-sets-$\delta$-distance with base sets being
$\{\lis{\o{x}}{,}{m}\}$. Then, 
$D(\o{x}+ \phim(\o{t}))$ 
has $(m(\ell - 1)+1)$-concentration, where 
$\ell = \log  (k+1) + \delta +1$.
\end{lemma}
\begin{proof}
We will prove the theorem by induction on $m$,
i.e.\ we first show the effect of shift in $\o{x}_m$ variables
and next invoke induction for the shift in $\cup_{i=1}^{m-1} \o{x}_i$ variables. 

{\em Base case:} When $m=1$, $\Dx = D(\o{x}_1)$ is simply a polynomial 
with $\delta$-distance.
So, $D(\o{x}_1+\phim(\o{t}_1))$ has $\ell$-concentration from Lemma~\ref{lem:dDistlConc}. 

{\em Induction hypothesis:} Assume that the statement is true for $m-1$ base sets. 

{\em Induction step:}
We will view the polynomial $\Dx$ as a polynomial 
on $\o{x}_m$ variables with coefficients coming 
from $\H_k(\F)[\lis{\o{x}}{,}{m-1}]$.
%Let $\o{x}_i = \{x_{i1}, x_{i2}, \dots, x_{in_i}\}$. 
Let $I_i$ denote the set of indices corresponding to the variables
in $\o{x}_i$.
%Recall that $I_m$ denotes the set of indices corresponding to the variables
%in $\o{x}_m$.
Then we can write $D = \sum_{S \subseteq I_m} u_S x_S$,
where $u_S \in \H_k(\F)[\lis{\o{x}}{,}{m-1}]$ for each $S \subseteq I_m$.
Let us define 
$$D' := 
D(\o{x}_1+\phim(\o{t}_1), \o{x}_2+\phim(\o{t}_2), \dots, \o{x}_m+\phim(\o{t}_m) ).$$
Also define 
$$\tilde{u}_S := u_S(\o{x}_1+\phim(\o{t}_1), \o{x}_2+\phim(\o{t}_2), \dots, \o{x}_{m-1}+\phim(\o{t}_{m-1}) )$$
and 
$$\tilde{D} := \sum_{S \subseteq I_m} \tilde{u}_S x_S.$$
Let $\tilde{\F} := \F(\lis{\o{x}}{,}{m-1},\lis{t}{,}{m-1})$. 
It is easy to see that $\tilde{D}$,
 as a polynomial in $\H_k(\tilde{\F})[\o{x}_m]$,
has $\delta$-distance. 
Hence, from Lemma~\ref{lem:dDistlConc}, 
$D' = \tilde{D}(\o{x}_m + \phim(\o{t}_m))$
has $\ell$-concentration over $\tilde{\F}(t_m)$. 
Note that our working field $\tilde{\F}$ having
variables $\{\lis{\o{x}}{,}{m-1} \}$ and $\{ \lis{t}{,}{m-1}\}$
is not a problem as $t_m$ is a new variable. 
$\ell$-Concentration of $D' = \sum_{S \subseteq I_m} u'_S x_S$ means 
that for any $T \subseteq I_m$,
\begin{equation}
u'_T \in \Span_{\tilde{\F}(t_m)}{\{ u'_S \mid S \subseteq I_m, \; \abs{S} < \ell \}}.
\label{eq:lconc}
\end{equation}
Next, we will show this kind of dependence 
among their coefficient vectors. 
Let us define a field 
$\F_t := \F(\lis{t}{,}{m})$.
Also define $\F' := \tilde{\F}(t_m) = \F_t(\lis{\o{x}}{,}{m-1})$.
Let $I_{[m-1]}$ denote the indices corresponding
to the variables in the set 
$\o{x} \setminus \o{x}_m$.
Now for any $T \subseteq I_m$, the vector 
$u'_T \in \H_k(\F')$ can be seen 
as a polynomial over $\o{x} \setminus \o{x}_m$ variables,
i.e.\ $u'_T = \sum_{U \subseteq I_{[m-1]}} u'_{T,U} x_U $,
where $u'_{T,U} \in \H_k(\F_t)$ for all $U \in I_{[m-1]}$. 
From Lemma~\ref{lem:polyToCoeffs},
we know that dependence among polynomial vectors
implies a dependence among their coefficients over the base field.
So, Equation~(\ref{eq:lconc})
together with Lemma~\ref{lem:polyToCoeffs}
implies that
for any $T \subseteq I_m$, $V \subseteq I_{[m-1]}$,

\begin{equation}
u'_{T,V} \in \Span_{\F_t} \{ u'_{S,U} \mid S \subseteq I_m, \; \abs{S} < \ell, \; U \subseteq I_{[m-1]}  \}.
\label{eq:lowSuppConc}
\end{equation}

Let us view $D'$ again as a polynomial in $\H_k(\F_t)[\lis{\o{x}}{,}{m}]$.
Equation~(\ref{eq:lowSuppConc}) shows that
in $D'$, the rank is concentrated on the coefficients 
corresponding to the monomials which have low-support from $\o{x}_m$.
Next, we will argue that the rank is
 actually concentrated on the coefficients corresponding to the monomials
which have low-support from all the $\o{x}_i$s. 

To show that, we will show low-support concentration in the circuit
$u'_S$ for any $S \in I_m$.
Let us define a polynomial 
$$\widehat{D} := D(\o{x}_1, \o{x}_2, \dots, \o{x}_{m-1}, \o{x}_m+\phim(\o{t}_m) ).$$
Viewing $\widehat{D}$ as a polynomial in $\o{x}_m$ variables,
we can write $\widehat{D} = \sum_{S \subseteq I_m} \hat{u}_S x_S$. 
We know that $D$ has $m$-base-sets-$\delta$-distance.
It is easy to see that shifting a polynomial in some variables 
preserves this property. 
Hence, $\widehat{D}$ has $m$-base-sets-$\delta$-distance
with base sets being $\lis{\o{x}}{,}{m}$.
Now, $\hat{u}_S \in \H_k(\F(t_m))[\lis{\o{x}}{,}{m-1}]$ is 
the coefficient of $x_S$ in $\widehat{D}$.
So, from Lemma~\ref{lem:recursive}, $\hat{u}_S$ 
will have $(m-1)$-base-sets-$\delta$-distance,
with base sets being $\lis{\o{x}}{,}{m-1}$. 
Now, by our induction hypothesis,
$\hat{u}_S$ will have low-support 
concentration after appropriate shifting. 
To be precise, 
$u'_S = \hat{u}_S(\o{x}_1+\phim(\o{t}_1), \o{x}_2+\phim(\o{t}_2), \dots, \o{x}_{m-1}+\phim(\o{t}_{m-1}))$ has $((m-1)(\ell - 1)+1)$-concentration.
This means that for any $S \subseteq I_m$, $V \subseteq I_{[m-1]}$, 
\begin{equation}
u'_{S,V} \in \Span_{\F_t}\{u'_{S,U} \mid U \subseteq I_{[m-1]}, \; \abs{U}<(m-1)(\ell - 1)+1 \}.
\label{eq:concD'_S}
\end{equation}
Combining Equation~(\ref{eq:lowSuppConc}) and (\ref{eq:concD'_S})
we get that for any $T \subseteq I_m$, $V \subseteq I_{[m-1]}$,
\begin{equation}
u'_{T,V} \in \Span_{\F_t} \{ u'_{S,U} \mid S \subseteq I_m, \; \abs{S} < \ell, \; U \subseteq I_{[m-1]}, \; \abs{U}<(m-1)(\ell - 1)+1  \}.
\label{eq:finalConc}
\end{equation}
Now, let us define $I_{[m]}:=I_m \cup I_{[m-1]}$. 
Viewing at $D'$ as a polynomial in variable set $\o{x}$
let us write $D' = \sum_{S \subseteq I_{[m]}} u'_S x_S$. 
From Equation~(\ref{eq:finalConc}), 
we can say that for any $T \subseteq I_{[m]}$, 
\begin{equation*}
u'_{T} \in \Span_{\F_t} \{ u'_{S} \mid S \subseteq I_{[m]}, \; \abs{S} < m(\ell - 1)+1\}.
\end{equation*}
Hence, $D'$ has $(m(\ell - 1)+1)$-concentration.
\end{proof}

%%%%%%%%%%%%%%%%%%%%%%%%%%%%%%%%%%%%%%%%%%%%%%%%%%%%%%%5
%%%%%%   ROABP appendix
%%%%%%%%%%%%%%%%%%%%%%%%%%%%%%%%%%%%%%%%%%%%%%%%%%%%%%%

\section{Building the Proof of Theorem~\ref{thm:ROABPHS}}
\label{app:ROABP}
Our first focus will be on
the matrix product $D(\o{x}) := \prod_{i=1}^{d} D_i$
which belongs to $\F^{w \times w}[\o{x}]$. 
%We will also view at this product
%as a polynomial in $\F^{w^2}[\o{x}]$.
%Like in Section $\ref{sec:deltaDistance}$,
%here also, 
We will show low-support concentration in $D(\o{x})$
over the matrix algebra $\F^{w \times w}$ (which is non-commutative!).
%as a polynomial in $\F^{w^2}[\o{x}]$.
%Our next lemma shows how does low-support concentration 
%imply hitting-sets. 

%\begin{lemma}
%If a polynomial $D(\o{x}) \in \F^{w \times w}[\o{x}]$ has $\ell$-concentration
%then there is a $n^{\ell}$ hitting-set for the polynomial $C(\o{x}) := \o{a}^T D(\o{x}) \o{b}$
%where $\o{a},\o{b} \in \F[\o{x}]^w$.
%\end{lemma}
%\begin{proof}
%Let $D(\o{x}) = \sum_{S \in [n]} D_S x_S$. 
%Then clearly, $C(\o{x}) = \sum_{S \in [n]} C_S x_S$,
%where $C_S = \o{a}^t D_S \o{b}$.
%$\ell$-conc means 
%\end{proof}

\subsection{Low Block-Support} 
Let the matrix product $D(\o{x}) := \prod_{i=1}^{d} D_i$
correspond to an ROABP such that
$D_i \in \F^{w \times w}[\o{x_i}]$ for all $i \in [d]$.
Let $n_i$ be the cardinality of $\o{x}_i$ and
let $n = \sum_{i=1}^d n_i$. 
For an exponent $e = (e_1, e_2, \dots, e_m) \in {\N}^m$, 
and for a set of variables $\o{y}= \{y_1,y_2, \dots, y_m\}$,
${\o{y}}^e$ will denote ${{y}_1}^{e_1} {y_2}^{e_2} \dots {y_m}^{e_m}$.

%Let $I_i$ denote the set of indices corresponding to the variables in 
%$\o{x}_i$ for all $ i \in [d]$
%and let $I := \cup_{i=1}^d I_i$.
Viewing $D_i$ as belonging to $\F^{w \times w}[{\o{x}}_i]$,
one can write $D_i := \sum_{e \in {\N}^{n_i}} D_{ie} {\o{x}}_i^e$.
In particular $D_{i {\bf 0}}$ refers to
the constant part of the polynomial $D_i$. 

For any $e \in \N^n$, 
support of the monomial $\o{x}^e$ is defined as 
$\Supp(e) := \{i \in [n] \mid e_i \neq 0 \}$
and support size is defined as $\supp(e) := \abs{\Supp(e)}$.
In this section, we will also define block-support
of a monomial. 
Any monomial $\o{x}^e$ for $e \in \N^n$, 
can be seen as a product $\prod_{i=1}^d {\o{x}_i}^{e_i}$,
where $e_i \in \N^{n_i}$ for all $i \in [d]$,
such that $e = (e_1, e_2, \dots, e_d)$.
We define {\em block-support} of $e$, $\bS(e)$  as
$\{i \in [d] \mid e_i \neq {\bf 0} \}$ and 
{\em block-support size} of $e$, $\bs(e) = \abs{\bS(e)}$.
 
Next, we will show low block-support concentration 
of $D(\o{x})$ when {\em each $D_{i {\bf 0}}$ is invertible}. 

As each $D_i$ is a polynomial over a different set of variables,
we can easily see that the coefficient of any 
monomial $\xe = \prod_{i=1}^{d} \xei$ in $D(\o{x})$
is 
\begin{equation}
\label{eq:coeffSplits}
D_e := \prod_{i=1}^d D_{ie_i}.
\end{equation}
Now, we will define a relation of {\em parent} and {\em children} 
between these coefficients. 

\begin{definition}
For $e^*,e \in \N^n$, $D_{e^*}$ is called a parent of $D_{e}$
if 
$\exists j \in [d]$,  
$j > \max \bS(e)$ or $j < \min  \bS(e) $,
such that
$\bS(e^*) = \bS(e) \cup \{j\}$
and $e^*_i = e_i$, $\forall i \in [d]$ with $i \neq j$.
\end{definition}

If $D_{e^*}$ is a parent of $D_{e}$ then 
$D_{e}$ is a {\em child} of $D_{e^*}$.
Note that any coefficient has at most two children,
on the other hand it can have many parents.
In the case when $j > \max \{ \bS(e)\}$
 we call $e$, the {\em left} child of $e^*$ and 
in the other case we call it the {\em right} child.

To motivate this definition, 
observe that if $j > \max  \bS(e)$ then by Equation~(\ref{eq:coeffSplits}) 
we can write
$D_{e^*} = D_e A^{-1} B$, where $A := \prod_{i=j}^{d} D_{i{\bf 0}}$
and $B := D_{j e^*_{j}}\prod_{i=j+1}^{d} D_{i{\bf 0}}$.
We will denote the product $A^{-1}B$ as $D_{e^{-1}e^*}$.
Similarly, if $j < \min \{ \bS(e) \}$ then one can write
$D_{e^*} = B A^{-1} D_e $, where $A := \prod_{i=1}^{j} D_{i{\bf 0}}$
and $B := \left(\prod_{i=1}^{j-1} D_{i{\bf 0}} \right) D_{j e^*_{j}}$.
In this case we will denote the product $B A^{-1}$ as $D_{e^*e^{-1}}$.
Note that the invertibility of $D_{i {\bf 0}}$s is crucial here. 

We also define {\em descendants} of a coefficient $D_e$
as $\descend(D_e) := \{D_f \mid f \in \N^n, \; \bS(f) \subset \bS(e) \}$.
%Observe that if $D_e$ is a descendant of $D_{e^*}$
%then $\max\{\bS(e^*)\} \geq \max\{\bS(e)\}$
%and $\min\{\bS(e^*)\} \leq \min\{\bS(e)\}$
%with one of inequalities being strict. 
Now, we will view the coefficients as $\F$-vectors 
%where $k = w^2$
and look at the linear dependence between them. 
The following lemma shows how these dependencies lift to the parent. 

\begin{lemma}[Child to parent]
Let $D_{e^*}$ be a parent of $D_{e}$. 
If $D_{e}$ is linearly dependent on its descendants,
then $D_{e^*}$ is linearly dependent on its descendants. 
\label{lem:parentchild}
\end{lemma} 
\begin{proof}
Let $D_{e}$ be the left child of $D_{e^*}$ 
(the other case is similar).
So, we can write 
\begin{equation}
D_{e^*} =  D_{e} D_{e^{-1}e^*}.
\label{eq:parentchild}
\end{equation}
Let the dependence of $D_e$ on its descendants be the following: 
$$
D_{e} = \sum_{\substack{f \\ \bS(f) \subset \bS(e)}} \alpha_f D_{f}.
$$
Using Equation~(\ref{eq:parentchild}) we can write,
$$
D_{e^*} = \sum_{\substack{f \\ \bS(f) \subset \bS(e)}} \alpha_f D_{f} D_{e^{-1}e^*}.
$$
Now, we just need to show that
for any $D_f$ with $\bS(f) \subset \bS(e)$,
 $D_{f} D_{e^{-1}e^*}$ 
is a valid coefficient of some monomial in $\Dx$ and 
also that it is a descendant of $D_{e^*}$.
Recall that $D_{e^{-1}e^*} = A^{-1}B$, 
where $A := \prod_{i=j}^{d} D_{i{\bf 0}}$
and $B := D_{j e^*_{j}}\prod_{i=j+1}^{d} D_{i{\bf 0}}$
and $\bS(e^*) = \bS(e) \cup \{j\}$.
We know that $j > \max\{\bS(e)\}$. 
Hence, $j > \max\{\bS(f)\}$ as $\bS(f) \subset \bS(e)$.
So, it is clear that $D_{f} D_{e^{-1}e^*}$
is the coefficient of ${\o{x}}^{f^*} := {\o{x}}^f {\o{x}_{j}}^{e^*_{j}}$.
It is easy to see that $\bS(f^*) = \bS(f) \cup \{j\} \subset \bS(e^*)$.
Hence, $D_{f^*}=D_{f} D_{e^{-1}e^*} $ is a descendant of $D_{e^*}$.

\end{proof}

Note that, the set of descendants of a coefficient 
strictly contains its children, grand-children, etc.
Clearly, if the descendants are more than $\dim_{\F}\F^{w \times w}$,
then there will be a linear dependence among them. So, 

\begin{lemma}
Any coefficient $D_e$, with $\bs(e) = w^2$, $\F$-linearly depends on
its descendants. 
\label{lem:blockSupportk}
\end{lemma}
\begin{proof}
First of all, we show that if a 
coefficient $D_{f^*}$ is nonzero then so are
its children. 
Let us consider its left child $D_f$ (the other case is similar).
Recall that we can write $D_{f^*} = D_f D_{f^{-1}f^*}$.
Hence if $D_f$ is zero, so is $D_{f^*}$. 

Let $k:= w^2$.
Now, consider a chain of coefficients
$D_{e^0}, D_{e^1}, \dots, D_{e^k}=D_e$,
such that for any $i \in [k]$, 
$D_{e^{i-1}}$ is a child of $D_{e^{i}}$. 
Clearly, $\bs(e^i) = i$ for $0 \leq i \leq k$. 
All the vectors in this chain are nonzero
because of our above argument, as $D_e$ is nonzero
(The case of $D_e = 0$ is trivial).
These $k+1$ vectors lie in $\F^{k}$, 
hence, there exists an $i \in [k]$
such that $D_{e^{i}}$ is linear dependent on 
$\{ D_{e^0}, \dots, D_{e^{i-1}}\}$. 
As descendants include children, grand-children, etc.,
we can say that $D_{e^i}$ is linearly dependent on
its descendants. 
Now, by applying Lemma~\ref{lem:parentchild}
repeatedly, we conclude $D_{e^{k}} = D_e$ 
is dependent on its descendants. 
\end{proof}

Note that, for a coefficient $D_e$ with $\bs(e) = i$,
its descendants have block-support strictly smaller than
$i$.
So, Lemma~\ref{lem:blockSupportk} means that coefficients
with block-support $w^2$ depend on coefficients with block-support
$\leq w^2-1$.
Now, we show $w^2$-block-support-concentration in $\Dx$,
i.e.\ any coefficient is   
dependent on the coefficients with block-support $\leq w^2-1$.

\begin{lemma}[$w^2$-Block-concentration]
\label{lem:bSConc}
Let $\Dx = \prod_{i=1}^{d} D_i (\o{x}_i) \in \F^{w \times w}[\o{x}]$
 be a polynomial
with $D_{i\bf{0}}$ being invertible for each $i \in [d]$.
Then $\Dx$ has $w^2$-block-support concentration. 
%For any coefficient $D_e$, 
%$$D_e \in \Span\{ D_f \mid f \in \N^n, \; \bs(f) \leq k-1 \}$$
\end{lemma} 
\begin{proof}
Let $k:= w^2$. We will actually show that
for any coefficient $D_e$ with $\bs(e) \geq k$ (the case when $\bs(e) < k$ is trivial), 
$$D_e \in \Span\{ D_f \mid f \in \N^n, \; \bS(f) \subset \bS(e) \text{ and } \bs(f) \leq k-1 \}.$$
We will prove the statement by induction on
the block-support of $D_e$, $\bs(e)$.

{Base case:} When $\bs(e)=k$, it has been already shown in 
Lemma~\ref{lem:blockSupportk}.

{\em Induction Hypothesis:} For any coefficient $D_e$
with $\bs(e) = i-1$ for $i-1 \geq k$, 
$$D_e \in \Span\{ D_f \mid f \in \N^n, \; \bS(f) \subset \bS(e) \text{ and } \bs(f) \leq k-1 \}.$$

{\em Induction step:}
Let us take a coefficient $D_e$ with $\bs(e) = i$. 
Consider any child of $D_e$, denoted by $D_{e'}$. 
As $\bs(e')=i-1$, by our induction hypothesis, 
$D_{e'}$ is linearly dependent on its descendants. 
So, from Lemma~\ref{lem:parentchild}, 
$D_{e}$ is linearly dependent on its descendants. 
In other words, 
\begin{equation}
D_e \in \Span\{ D_f \mid \bS(f) \subset \bS(e) \text{ and } \bs(f) \leq i-1 \}.
\label{eq:De}
\end{equation}
Again, by our induction hypothesis, for any coefficient $D_f$, 
with $\bs(f) \leq i-1$, 
\begin{equation}
D_f \in \Span\{ D_g \mid \bS(g) \subset \bS(f) \text{ and } \bs(g) \leq k-1 \}.
\label{eq:Df}
\end{equation}
Combining Equations (\ref{eq:De}) and (\ref{eq:Df}),
we get
$$D_e \in \Span\{ D_g \mid \bS(g) \subset \bS(e) \text{ and } \bs(g) \leq k-1 \}.$$
\end{proof}

Now, we show low block-support concentration in
the actual polynomial computed by an ROABP, i.e.\
in $C(\o{x}) = D_0^T (\prod_{i=1}^d D_i) D_{d+1}$, 
where $D_0,D_{d+1} \in F^w[\o{x}]$.

\begin{corollary}
Let $\o{x} = \o{x}_0 \sqcup \o{x}_1 \sqcup \dotsm \sqcup \o{x}_{d+1}$.
Let $D(\o{x}) \in \F^{w \times w}[\o{x}_1, \dots, \o{x}_d]$
 be a polynomial described in Lemma~\ref{lem:bSConc}.
Let $C(\o{x}) = D_0^T D D_{d+1} \in \F[\o{x}]$
 be a polynomial
with $D_0 \in \F^w[\o{x}_0]$, $D_{d+1} \in \F^w[\o{x}_{d+1}]$.
%$D_i \in \F^{w \times w}[\o{x}_i]$
%and $D_{i\bf{0}}$ being invertible for each $i \in [d]$.
Then $C(\o{x})$ has $(w^2+2)$-block-support concentration.  
\label{cor:bSConc}
\end{corollary}
\begin{proof}
Let $k:=w^2$.
Lemma~\ref{lem:bSConc} shows that $\Dx$
% = \prod_{i=1}^{d} D_i (\o{x}_i)$
 has $k$-block-support concentration. 
The coefficient of $\o{x}^e$ in $C$ 
is $C_e := D_{0e_0} \prod_{i=1}^d D_{ie_i} D_{(d+1)e_{d+1}}$,
where $e = (e_0, e_1, \dots,e_d, e_{d+1})$. 
Let $D_e := \prod_{i=1}^d D_{ie_i}$.
By $k$-block-support concentration of $\Dx$,
$$D_e \in \Span\{ D_f \mid \bs(f) \leq k-1 \}.$$
Which implies, 
$$C_e \in \Span\{ D_{0e_0} D_f D_{(d+1)e_{d+1}} \mid \bs(f) \leq k-1 \}.$$ 
Clearly, $D_{0e_0} D_f D_{(d+1)e_{d+1}}$ is the coefficient of the monomial 
$x_0^{e_0} x_1^{f_1} \dotsm x_d^{f_d} x_{d+1}^{e_{d+1}}$. 
Hence, $C_e \in \Span\{  C_f \mid \bs(f) \leq k+1 \}$.
\end{proof}

\subsection{Low-support concentration}
Now, we argue that if $D(\o{x}) = \prod_{i=1}^{d} D_i (\o{x}_i)$ has
low block-support concentration and moreover if each $D_i$ has 
low-support concentration then $\Dx$ has an appropriate low-support concentration. 

\begin{lemma}[Composition]
If a polynomial 
$D(\o{x}) = \prod_{i=1}^{d} D_i (\o{x}_i) \in \F^{w \times w}[\o{x}]$ 
has $\ell$-block-support concentration and
$D_i(\o{x}_i)$ has $\ell'$-support concentration for all $i \in [d]$ then 
$\Dx$ has $\ell \ell'$-support concentration.  
\label{lem:ll'}
\end{lemma}
\begin{proof}
Recall that as $D_i$'s are polynomials over disjoint sets of variables,
any coefficient $D_f$ in $\Dx$ can be written as
$\prod_{i=1}^d D_{if_i}$, where $f = (f_1, f_2, \dots, f_d)$ and
$D_{if_i}$ is the coefficient corresponding to the monomial $\o{x}_i^{f_i}$
in $D_i$.
From the definition of $\bS(f)$, we know that
$f_i = 0$, for any $i \notin \bS(f)$.
%We can also write 
%$D_f = \prod_{i \in \bs(f)} D_{if_i} \prod_{i \not\in \bs(f)} D_{i\bf{0}}$. 
From $\ell'$-support concentration of $D_i(\o{x}_i)$, 
we know that for any coefficient $D_{if_i}$,
$$D_{if_i} \in \Span \{D_{ig_i} \mid g_i \in \N^{n_i}, \; \supp(g_i) \leq \ell'-1 \}.$$
Using this, we can write
%$$D_f \in \Span \{ \prod_{i \in \bs(f)} D_{ig_i} \prod_{i \not\in \bs(f)} D_{i\bf{0}} \mid \forall i \in \bs(f), \; g_i \in \N^{n_i} \text{ and } \supp(g_i) \leq \ell'-1  \} $$
$$D_f \in \Span \left\{ \prod_{i = 1}^d D_{ig_i} \mid  g_i \in \N^{n_i}, \; \supp(g_i) \leq \ell'-1, \; \forall i \in [d] \text{ and } g_i={\bf 0}, \; \forall i \notin \bS(f)  \right\} .$$
Note that the product $\prod_{i \in [d]} D_{ig_i}$ 
will be the coefficient of a monomial $\o{x}^g$ such that $\bS(g) \subseteq \bS(f)$
because $g_i={\bf 0}, \; \forall i \notin \bS(f)$.
Clearly, if $\supp(g_i) \leq \ell'-1, \; \forall i \in \bS(f)$ then $\supp(g) \leq (\ell'-1)\bs(f)$.
So, one can write
\begin{equation}
D_f \in \Span\{ D_g \mid g \in \N^n, \; \supp(g) \leq (\ell'-1)\bs(f)  \} .
\label{eq:span-bsf}
\end{equation}
From $\ell$-block-support concentration of $\Dx$,
 we know that for any 
coefficient $D_e$ of $\Dx$, 
\begin{equation}
D_e \in \Span\{ D_f \mid f \in \N^n, \; \bs(f) \leq \ell - 1 \}.
\label{eq:span-l-1}
\end{equation}
Using Equations (\ref{eq:span-bsf}) and (\ref{eq:span-l-1}),
we can write for any 
coefficient $D_e$ of $\Dx$, 
$$D_e \in \Span\{ D_g \mid g \in \N^n, \; \supp(g) \leq (\ell'-1)(\ell - 1) \}.$$
Hence, $\Dx$ has $((\ell - 1)(\ell'-1)+1)$-support concentration and
hence $\ell \ell'$-support concentration.
\end{proof}

Now, we just need to show low-support concentration 
of each $D_i$. 
To achieve that we will use some efficient shift.
Shifting will serve a dual purpose. 
Recall that for Lemma~\ref{lem:bSConc}, we need
invertibility of the constant term in $D_i$, i.e.\ $D_{i\bf{0}}$, 
for all $i \in [d]$.
In case $D_{i \bf{0}}$ is not invertible for some $i \in [d]$, 
after a shift it might become invertible, 
since $D_i$ is assumed invertible in the sparse-invertible model. 
For the shifted polynomial $D'_i(\o{x}_i) := D_i(\o{x}_i + \phi(\o{t}_i))$,
its constant term $D'_{i\bf{0}}$ is just an evaluation of $D_i(\o{x})$, 
i.e.\ $D_i|_{\o{x}_i = \phi(\o{t}_i)}$.
%Hence, if $\det(D_i({\o{x}}_i)) = 0$ (viewing $D_i({\o{x}}_i)$ as 
%an element in $(\F[\o{x}_i])^{w \times w}$),
% then $\det(D'_{i{\bf 0}}) = 0$.
%This means that if $\det(D_i({\o{x}}_i)) = 0$, then even after shifting, $D'_{i{\bf 0}}$
%cannot become invertible. 
%So, we have to assume that $D_i(\o{x}_i)$ is an invertible matrix
%for all $i \in [d]$. 
%In other words, the product $D(\o{x}) = \prod_{i=1}^d D_i(\o{x}_i)$ is invertible. 
Now, we want a shift for $D_i$ which would ensure that 
$\det(D'_{i{\bf 0}}) \neq 0$ and that $D'_i$ has low-support concentration. 
For both the goals we use the sparsity of the polynomial. 

For a polynomial $D$, let its sparsity set 
$\Sp(D)$ be the set of monomials in $D$
with nonzero coefficients 
and $\sp(D)$ be its
sparsity, i.e.\ $\sp(D) = \abs{\Sp(D)}$.
Let,
for a polynomial $D(\o{x}) \in \F^{w \times w}[\o{x}]$,
$S = \Sp(D)$ and $s = \abs{S}$. 
Then it is easy to see that 
for its determinant polynomial 
$\Sp(\det(D)) \subseteq S^w$,
where $S^w := \{ m_1 m_2 \dotsm m_w \mid m_i \in S, \; \forall i \in [w]   \}$.
Hence $\sp(\det(D)) \leq s^w$.
Now, suppose $\det(D) \neq 0$. 
We will describe an efficient shift which will make the
constant term, of the shifted polynomial, invertible. 
Let $\phi \colon \o{t} \to \{t^i\}_{i=0}^{\infty}$ be a monomial map
which separates all the monomials in $\det(D(\o{t}))$,
i.e.\
for any two $\o{t}^{e_1},{\o{t}}^{e_2} \in \Sp(\det(D(\o{t})))$,
$\phi(\o{t}^{e_1}) \neq \phi(\o{t}^{e_2})$. 
It is easy to see that if we shift each $x_i$ by $\phi(t_i)$
to get $D'(\o{x}) = D(\o{x} + \phi(\o{t}))$ then 
$\det(D'_{i{\bf 0}}) = \det(D|_{\o{x} = \phi(\o{t})}) \neq 0$.

For sparse polynomials, Agrawal et al.\ \cite[Lemma 16]{ASS13} have 
given an efficient shift to achieve low-support concentration. 
Here, we rewrite their lemma. % in some different words.
The map $\phi_{\ell'} \colon \o{t} \to \{t^i\}_{i=0}^{\infty}$ is 
said to be separating $\ell'$-support monomials of degree $\delta$,
if for any two monomials $\o{t}^{e_1}$ and $\o{t}^{e_2}$ which have
support bounded by $\ell'$ and degree bounded by $\delta$, 
$\phi_{\ell'}(\o{t}^{e_1}) \neq \phi_{\ell'}(\o{t}^{e_2})$. 
For a polynomial $\Dx$, let $\mu(D)$ be the maximum support 
of a monomial in $D$, i.e.\ $\mu(D) := \displaystyle\max_{\o{x}^e \in \Sp(D)} \supp(e)$.

\begin{lemma}[\cite{ASS13}]
\label{lem:sparse}
Let $V$ be a $\F$-vector space of dimension $k$.
Let $D(\o{x}) \in V[\o{x}]$ be a polynomial with
degree bound $\delta$. 
Let $\ell := 1 + 2 \min\{ \ceil{\log (k \cdot \sp(D))}, \mu(D) \}$ and $\phi_{\ell}$ be a
monomial map separating $\ell$-support monomials of degree $\delta$.
Then $D(\o{x}+\phi_{\ell}(\o{t}))$ has $\ell$-concentration
over $\F(t)$.
\end{lemma} 

The \cite{ASS13} version of the Lemma~\ref{lem:sparse} gave
 a concentration result
about sparse polynomials over $\H_k(\F)$. But observe that
the process of shifting and the definition of concentration 
only deal with the additive structure of $\H_k(\F)$,
and the multiplication structure is irrelevant. Hence,
the result is true over any
$\F$-vector space,
%$V$-algebra of dimension $k$,
in particular,
over the matrix algebra.  
By combining these observations, we have the following. 

\begin{lemma}
\label{lem:ROABPconc}
Let $D(\o{x}) = \prod_{i=1}^{d} D_i(\o{x}_i)$ be a polynomial 
in $\F^{w \times w}[\o{x}]$ with $\det(D) \neq 0$
such that for all $i \in [d]$,
 $D_i$ has degree bounded
by $\delta$, $\sp(D_i) \leq s$ and $\mu(D_i) \leq \mu$.
Let $\ell := 1 + 2 \min\{ \ceil{\log (w^2 \cdot s)}, \mu \}$
and $M :=  \poly(s^w (n \delta)^{\ell})$.
Then there is a set of $M$ monomial maps with degree bounded by
$M \log M$ 
such that for at least one of the maps $\phi$,
%\colon t_i \to t^{b_i}$,
$D' := D(\o{x}+\phi(\o{t}))$ has $\ell w^2$-concentration. 
\end{lemma}
\begin{proof}
Let $\phi \colon \o{t} \to \{t^i\}_{i=0}^{\infty}$ be a map 
such that it separates all the monomials in $\Sp(\det(D_i(\o{t}_i)))$,
for all $i \in [d]$. There are $ds^{2w}$ such monomial pairs. 
Also assume that $\phi$ separates all monomials of support bounded by $\ell$. 
There are $(n \delta)^{O(\ell)}$ such monomials. 
Hence, total number of 
monomial pairs which need to be separated are $s^{O(w)} + (n \delta)^{O(\ell)}$.
%where $a =O(ds^{2w} + n^{2 \ell} \delta^{2 \ell})$.
From Lemma~\ref{lem:kronecker}, 
we know that there is a set of $M$ monomial maps
with highest degree $M\log M$ such that at least one of the maps
$\phi$ separates the desired monomials, 
%$M = n a \log \delta $.
where $M = \poly(s^w (n \delta)^{\ell})$.
As the map $\phi$ separates all the monomials in $\Sp(\det(D_i(\o{t}_i)))$,
$\det(D_i(\phi(\o{t}_i))) \neq 0$ and hence, 
$D'_{i{\bf 0}}$ is invertible for all $i \in [d]$.
So, $D'(\o{x})$ has $w^2$-block-support concentration from Lemma~\ref{lem:bSConc}.

From Lemma~\ref{lem:sparse}, $D'_i(\o{x}_i)$ has $\ell$-concentration for all 
$i \in [d]$.
Hence, from Lemma~\ref{lem:ll'}, $D'(\o{x})$ has $\ell w^2$-concentration.  
\end{proof}

%\begin{proof}[Proof of Theorem~\ref{thm:ROABPHS}(restated)]
Now, we come back to the proof of Theorem~\ref{thm:ROABPHS} (restated in Section~\ref{sec:ROABP}).
We want to find a hitting set for $C(\o{x}) = D_0^T D D_{d+1}$, where
$D \in \F^{w \times w}[\o{x}]$ is the polynomial as described in Lemma~\ref{lem:ROABPconc} and 
$D_0 \in \F^w[\o{x}_0]$, $D_{d+1} \in \F^w[\o{x}_{d+1}]$.
%Let $g(\o{x})$ be the polynomial described in the theorem. 
Using $(w^2+2)$-block-support concentration 
of $C(\o{x})$ from Corollary \ref{cor:bSConc},
and arguing as in the proof of Lemma~\ref{lem:ROABPconc}, 
we show 
$\ell (w^2+2)$-concentration of $C(\o{x}+\phi(\o{t}))$,
where $\phi$ is a map which needs to separate $s^{O(w)} + (n \delta)^{O(\ell)}$-many
monomials pairs. 
From Lemma~\ref{lem:kronecker}, there is set of $\poly(s^w (n \delta)^{\ell})$-many
shifting maps with highest degree $\poly(s^w (n \delta)^{\ell})$ such that one of them
is the desired map.
Similar to Lemma~\ref{lem:hsFromlConc}, we can show that 
if $C(\o{x}+\phi(\o{t}))$ has $\ell (w^2+2)$-concentration and
has degree bound $\delta^{O(1)}$ then there is a hitting-set 
of size $(n \delta)^{O(\ell w^2)}$.
Each of these evaluations will be a polynomial in $t$ with highest degree $\poly(s^w (n \delta)^{\ell})$.
Hence, total time complexity becomes $\poly(s^w(n \delta)^{\ell w^2})$. 
%Hence, from Lemma~\ref{lem:kronecker} and 
%Lemma~\ref{lem:finalHS}, there is a hitting-set of size
%$O((M^2\log M) n^{\ell (w^2+2)} ) $, where 
%$M =  n \log \delta(d s^{2w} + (n \delta)^{2 \ell})$.
%$\ell(w^2+2)$-concentration simply means that 
%if $g(\o{x}) \neq 0$ then one of its coefficients
%with support smaller than $\ell(w^2+2)$, is nonzero.
%There are $O(n^{\ell (w^2+2)})$ such coefficients. 
%Hence, we get a hitting-set of size $O((M^2 \log M) n^{\ell (w^2+2)})$ 
%from Lemma~\ref{}.
%\end{proof}

Note that for constant width ROABP,
when $\mu(D_i)$ is bounded by a constant for each $0 \leq i \leq d+1$,
in particular when each $D_i$ is univariate, 
the parameter $\ell$ becomes constant and
the hitting-set becomes polynomial-time.

\section{Width-$2$ Read Once ABP}
\label{app:2ROABP}
In Section~\ref{sec:ROABP}, the crucial part in finding a hitting-set for an ROABP,
is the assumption that the matrix product $\Dx$ is invertible. 
Now, we will show that for width-$2$ ROABP this assumption is not required.
%Let us say the polynomial $D_0^T D D_{d+1}$, for $D_0,D_{d+1} \in \F^w[\o{x}]$
%and $D \in \F^{w \times w}[\o{x}]$, corresponds to 
%an {\em invertible ROABP} if the matrix $D$ is invertible. 
Via a factorization property of $2 \times 2$ matrices,
 we will show that PIT for width-$2$ sparse-factor ROABP
reduces to PIT for width-$2$ sparse-invertible-factor ROABP. 
%Consider the following lemma.

\begin{lemma}[$2 \times 2$ invertibility]
\label{lem:width2}
Let $C(\o{x}) =  D_0^T \left( \prod_{i=1}^d D_i  \right) D_{d+1}$ be a polynomial computed
by a width-$2$ sparse-factor ROABP.
Then we can write $\alpha(\o{x}) C(\o{x}) = C_1(\o{x})C_2(\o{x}) \dotsm C_{m+1}(\o{x})$,
for some nonzero $\alpha \in \F[\o{x}]$ and some $m \leq d$, 
where $C_i(\o{x})$ is a polynomial computed by a width-$2$ sparse-invertible-factor 
ROABP, for all $i \in [m+1]$. 
\end{lemma}
\begin{proof}
Let us say, for some $i \in [d]$, $D_i(\o{x}_i)$ is not invertible. 
Let $D_i = \left[ \begin{smallmatrix} a_i & b_i \\ c_i & d_i \end{smallmatrix} \right]$ with
$a_i,b_i,c_i,d_i \in \F[\o{x}_i]$ and $a_i d_i = b_i c_i$.
Without loss of generality, at least one of $\{a_i,b_i,c_i,d_i\}$ is nonzero. 
Let us say $a_i \neq 0$ (other cases are similar).
Then we can write,
$$ \begin{bmatrix} a_i & b_i \\ c_i & d_i \end{bmatrix} 
= \frac{1}{a_i} \begin{bmatrix} a_i \\ c_i  \end{bmatrix}
\begin{bmatrix} a_i & b_i \end{bmatrix} . 
$$ 
In other words, we can write
$\alpha_i D_i = A_i B_i^T $, where $A_i,B_i \in \F^2[\o{x}_i]$ and 
$0 \neq \alpha_i \in \{a_i,b_i,c_i,d_i\} $.
Note that $\sp(\alpha_i), \sp(A_i),\sp(B_i) \leq \sp(D_i)$.
Let us say the set of non-invertible $D_i$s is $\{D_{i_1}, D_{i_2}, \dots, D_{i_{m}} \}$.
Writing all of them in the above form we get, 
$$C(\o{x}) \prod_{j=1}^m \alpha_{i_j} = 
\prod_{j=1}^{m+1} C_j ,
$$
where 
\begin{equation*}
C_j := \begin{cases}
D_0^T \left( \prod_{i=1}^{i_1-1} D_i \right) A_{i_1} & \text{ if } j= 1, \\
 B_{i_{j-1}}^T \left(\prod_{i=i_{j-1}+1}^{i_{j}-1} D_i \right) A_{i_{j}}
& \text{ if } 2 \leq j \leq m, \\
B_{i_m}^T \left( \prod_{i=i_m+1}^{d} D_i \right) D_{d+1} &  \text{ if } j = m+1. 
\end{cases}
\end{equation*}
Clearly, for all $j \in [m+1]$, $C_j$ can be computed by a sparse-invertible-factor ROABP.
\end{proof}

Now, from the above lemma it is easy to construct a hitting-set. 
First we show a general result about hitting-sets
for a product of polynomials from some class.
\begin{lemma}[Lagrange interpolation]
\label{lem:lagrange}
Suppose $\Hit$ is a hitting-set for a class of polynomials $\C$.
Let $C(\o{x}) = C_1(\o{x}) C_2(\o{x}) \dotsm C_m(\o{x})$,
where $C_i \in \C$ and has degree bounded by $\delta$, for all $i \in [m]$.
There is a hitting-set of size $m \delta \abs{\Hit} +1$ for $C(\o{x})$.
\end{lemma}
\begin{proof}
Let $h = \abs{\Hit}$  and 
$\Hit = \{ \o{\alpha}_1, \o{\alpha}_2, \dots, \o{\alpha}_{h} \}$.
Let $B := \{\beta_i\}_{i=1}^{h}$ be a set of constants.
The Lagrange interpolation $\o{\alpha}(u)$ of the points in $\Hit$ is defined
as follows 
$$ \o{\alpha}(u) := \sum_{i=1}^{h} \frac{\prod_{j \neq i} (u - \beta_j) }{\prod_{j \neq i} (\beta_i - \beta_j) } \o{\alpha}_i . 
$$
The key property of the interpolation is that 
when we put $u = \beta_i$,
$\o{\alpha}(\beta_i) = \o{\alpha}_i$ for all $i \in [h]$.
For any $a \in [m]$, we know that 
$C_a(\o{\alpha}_i) \neq 0$, for some $i \in [h]$. 
Hence, $C_a(\o{\alpha}(u))$ as a polynomial in $u$ is nonzero 
because $C_a(\o{\alpha}(\beta_i)) = C_a(\o{\alpha}_i) \neq 0$.
So, we can say $C(\o{\alpha}(u)) \neq 0$ as a polynomial in $u$. 
Degree of $\o{\alpha}(u)$ is $h$. So, degree of $C(\o{\alpha}(u))$
in $u$ is bounded by $m \delta h$. 
We can put $(m \delta h +1)$-many distinct values of $u$ 
to get a hitting-set for $C(\o{\alpha}(u))$.
\end{proof}

Note that a hitting-set for $\alpha(\o{x}) C(\o{x})$ is also 
a hitting-set for $C(\o{x})$ if $\alpha$ is a nonzero polynomial.
Recall that we get a hitting-set for invertible ROABP from
 Theorem~\ref{thm:ROABPHS} (Section~\ref{sec:ROABP}). 
Lemma~\ref{lem:width2} tells us how to 
write a width-$2$ ROABP as a product of width-$2$ invertible ROABPs. 
Combining these results with Lemma~\ref{lem:lagrange}
we directly get the following. 

\begin{theorem}
\label{thm:ROABP22HS}
Let $C(\o{x}) = D_0^T(\o{x}_0) (\prod_{i=1}^{d} D_i(\o{x}_i)) D_{d+1}(\o{x}_{d+1})$
be a polynomial 
in $\F[\o{x}]$ computed by a width-$2$ ROABP
 such that for all $0 \leq i \leq d+1$,
 $D_i$ has degree bounded
by $\delta$, $\sp(D_i) \leq s$ and $\mu(D_i) \leq \mu$.
Let $\ell := 1 + 2 \min\{ \ceil{\log (4 \cdot s)}, \mu \}$
%and $M :=  n \log \delta(d s^{4} + (n \delta)^{2 \ell})$.
Then there is a hitting-set of size
%$O((d+1) \delta M^2 \log M n^{6 \ell })$ for $C(\o{x})$. 
$\poly((n \delta s)^{\ell})$.
\end{theorem} 

We remark again that when %$\mu$ is constant or in particular when 
all $D_i$s are constant-variate or linear polynomials,
the hitting-set is polynomial-time.

\section{Sum of set-multilinear circuits}
\label{sec:sumSetMult}
In this section, we will reduce the PIT for 
sum of constantly many set-multilinear depth-$3$ circuits, 
to the PIT for depth-$3$ circuits with $m$ base sets
having $\delta$ distance, where $m \delta = o(n)$. 
Thus, we get a subexponential time whitebox algorithm 
for this class (from Theorem~\ref{thm:baseSetsHS}). 
Note that a sum of constantly many set-multilinear depth-$3$ circuits
is equivalent to a depth-$3$ multilinear circuit such that 
the number of distinct partitions, induced by its product gates,
is constant.

We first look at the case of two partitions. 
For a partition $\P$ of $[n]$, 
let $\P|_B$ denote the restriction of $\P$ on a base set
$B \subseteq [n]$. 
For example, if $\P = \{ \{1,2\}, \{3,4\}, \{5,6, \dots, n\}\}$ and $B = \{1,3,4\}$
then $\P|_B = \{ \{1\}, \{3,4\} \}$.
Recall that $d(\P_1,\P_2, \dots, \P_c)$ denotes the {\em distance}
of the partition sequence $(\P_1,\P_2, \dots, \P_c)$ (Definition~\ref{def:distance}).  
For a partition sequence $(\P_1, \P_2, \dots \P_c)$, 
and a base set $B \subseteq [n]$, 
let $d_B(\P_1, \P_2, \dots, \P_c)$ denote the distance of the partition sequence
when restricted to the base set $B$, i.e.\
$d(\P_1|_B, \P_2|_B, \dots, \P_c|_B)$.

\begin{lemma}
\label{lem:twoPartitions}
For any two partitions $\{ \P_1, \P_2 \}$ of the set $[n] $,
there exists a partition of $[n]$,
into at most $2 \sqrt{n}$ base sets $\{B_1, B_2, \dots, B_{m} \}$
 $(m < 2 \sqrt{n})$,
such that for any $i \in [m]$,
either $d_{B_i}(\P_1, \P_2) = 1$ or $d_{B_i}(\P_2, \P_1) =1$.
\end{lemma}          
\begin{proof}
Let us divide the set of colors in the partition $\P_1$, into 
two types of colors: One with at least $\sqrt{n}$ elements and the other with less than
$\sqrt{n}$ elements. 
In other words, $\P_1 = \{X_1, X_2, \dots, X_{r} \} \cup \{ Y_1, Y_2, \dots, Y_{q}\}$
such that $\abs{X_i} \geq \sqrt{n}$ and $\abs{Y_j} < \sqrt{n}$,
 for all $i \in [r], \; j \in [q]$.
Let us make each $X_i$ a base set, i.e.\ $B_i = X_i$, $\forall i \in [r]$. 
As $\abs{X_i} \geq \sqrt{n}, \; \forall i \in [r]$, we get $r \leq \sqrt{n}$.
Now, for any $i \in [r]$, $\P_1|_{B_i}$ has only one color.
Hence, irrespective of what colors $\P_2|_{B_i}$ has, 
$d_{B_i}(\P_2, \P_1) = 1$, for all $i \in [r]$.
% (it follows from the definition of distance).

Now, for the other kind of colors, we will make base sets  
which have exactly one element from each color $Y_j$.
More formally, let $Y_j = \{y_{j,1}, y_{j,2}, \dots, y_{j, r_j} \}$, for all $j \in [q]$.
Let $r' = \max\{ r_1, r_2, \dots, r_q\}$ ($r' < \sqrt{n}$). 
Now define base sets $B'_1, B'_2, \dots, B'_{r'}$ such that
for any $a \in [r']$, $B'_a = \{y_{j,a} \mid j \in [q], \; \abs{Y_j} \geq a \} $. 
In other words, all those $Y_j$s which have at least $a$ elements, contribute
their $a$-th element to $B'_a$.
Now for any $a \in [r']$, $\P_1|_{B'_a} = \{ \{y_{j,a}\} \mid j \in [q], \; \abs{Y_j} \geq a \}$,
i.e.\ it has exactly one element in each color. 
Clearly, irrespective of what colors $\P_2|_{B'_a}$ has, 
$d_{B'_a} (\P_1,\P_2) = 1$, for all $a \in [r']$.

$\{B_1, B_2, \dots, B_r\} \cup \{B'_1, B'_2, \dots, B'_{r'}\}$ is our final set of base sets. 
Clearly,  they form a partition of $[n]$. 
The total number of base sets, $m = r + r' < 2 \sqrt{n}$.

\end{proof}

Now, we generalize this to any constant number of partitions, by induction. 

\begin{lemma}
\label{lem:cPartitions}
For any $c$ partitions $\{ \P_1, \P_2, \dots, \P_c \}$ of the set $[n]$,
there exists a partition of $[n]$,
into $m$ base sets $\{B_1, B_2, \dots, B_{m} \}$
with $m < 2^{c-1} \cdot n^{1 - (1/2^{c-1})}$
such that for any $i \in [m]$,
there exists a permutation of the partitions, 
$(\P_{i_1} , \P_{i_2}, \dots, \P_{i_c})$
with $d_{B_i} (\P_{i_1} , \P_{i_2}, \dots, \P_{i_c}) = 1$.
\end{lemma}
\begin{proof}
Let $f(c, n) := 2^{c-1} \cdot n^{1 - (1/2^{c-1})} $.
The proof is by induction on the number of partitions. 

{\em Base case:} For $c = 2$, $f(c,n)$ becomes 
$2 \sqrt{n}$. Hence, the statement follows from Lemma~\ref{lem:twoPartitions}.

{\em Induction hypothesis:} The statement is true for any $c-1$ partitions. 
 
{\em Induction step:} Like in Lemma~\ref{lem:twoPartitions}, we divide the 
set of colors in $\P_1$ into two types of colors. 
Let $\P_1 = \{X_1, X_2, \dots, X_{r} \} \cup \{ Y_1, Y_2, \dots, Y_{q}\}$
such that $\abs{X_i} \geq \sqrt{n}$ and $\abs{Y_j} < \sqrt{n}$,
 for all $i \in [r], \; j \in [q]$.
Let us set $B_i = X_i$ and let $n_i := \abs{B_i}$, $\forall i \in [r]$ . 
%As $\abs{X_i} \geq \sqrt{n}, \; \forall i \in [r]$, we get $r \leq \sqrt{n}$.
Our base sets will be further subsets of these $B_i$s. 
For a fixed $i \in [r]$, let us define $\P'_h = \P_h|_{B_i} $, as a partition of the set $B_i$,
for all $h \in [c]$.
Clearly, $\P'_1$ has only one color. 
Now, we focus on the partition sequence $(\P'_2, \P'_3, \dots, \P'_c)$.
From the inductive hypothesis, 
there exists a partition of $B_i$ into $m_i$ base sets
$\{ B_{i,1}, B_{i,2}, \dots, B_{i,m_i} \}$ ($m_i \leq f(c-1, n_i)$) such that 
for any $u \in [m_i]$, there exists a permutation of 
$(\P'_2, \P'_3, \dots, \P'_c)$, given by $(\P'_{i_2}, \P'_{i_3}, \dots, \P'_{i_c})$,  
with $d_{B_{i,u}} (\P'_{i_2}, \P'_{i_3}, \dots, \P'_{i_c}) = 1$.
As $\P'_1$ has only one color, so does $\P'_1|_{B_{i,u}}$. 
Hence, $d_{B_{i,u}} (\P'_{i_2}, \P'_{i_3}, \dots, \P'_{i_c}, \P'_1)$ is also $1$.
From this, we easily get $d_{B_{i,u}} (\P_{i_2}, \P_{i_3}, \dots, \P_{i_c}, \P_1) = 1$.
The above argument can be made for all $i \in [r]$.

Now, for the other colors, 
we proceed as in Lemma~\ref{lem:twoPartitions}.
Let $Y_j = \{y_{j,1}, y_{j,2}, \dots, y_{j, r_j} \}$, for all $j \in [q]$.
Let $r' = \max\{ r_1, r_2, \dots, r_q\}$ ($r' < \sqrt{n}$). 
Now define sets $B'_1, B'_2, \dots, B'_{r'}$ such that
for any $a \in [r']$, $B'_a = \{y_{j,a} \mid j \in [q], \; \abs{Y_j} \geq a \} $. 
In other words, all those $Y_j$s which have at least $a$ elements, contribute
their $a$-th element to $B'_a$.
Let $n'_a := \abs{B'_a}$, for all $a \in [r']$.
Our base sets will be further subsets of these $B'_a$s. 
For a fixed $a \in [r']$, let us define $\P'_h = \P_h|_{B'_a} $, 
as a partition of the set $B'_a$,
for all $h \in [c]$.
Clearly, $\P'_1$ has exactly one element in each of its colors. 
Now, we focus on the partition sequence $(\P'_2, \P'_3, \dots, \P'_c)$.
From the inductive hypothesis, 
there exists a partition of $B'_a$ into $m'_a$ base sets
$\{ B'_{a,1}, B'_{a,2}, \dots, B'_{a,m'_a} \}$ ($m'_a \leq f(c-1, n'_a)$) such that 
for any $u \in [m'_a]$, there exists a permutation of 
$(\P'_2, \P'_3, \dots, \P'_c)$, given by $(\P'_{i_2}, \P'_{i_3}, \dots, \P'_{i_c})$,  
with $d_{B'_{a,u}} (\P'_{i_2}, \P'_{i_3}, \dots, \P'_{i_c}) = 1$.
As $\P'_1$ has exactly one element in each of its colors, so does $\P'_1|_{B'_{a,u}}$. 
Hence, $d_{B'_{a,u}} (\P'_1, \P'_{i_2}, \P'_{i_3}, \dots, \P'_{i_c})$ is also $1$.
From this, we easily get $d_{B'_{a,u}} (\P_1, \P_{i_2}, \P_{i_3}, \dots, \P_{i_c}) = 1$.
The above argument can be made for all $a \in [r']$.

Our final set of base sets will be $\{B_{i,u} \mid i \in [r], \; u \in [m_i] \} \cup
\{ B'_{a,u} \mid a \in [r'], \; u \in [m'_a] \}$. 
As argued above, when restricted to any of these base sets, the given partitions 
have a sequence, which has distance $1$. 
Now, we need to bound the number of these base sets, 
$$m = \sum_{i \in [r]} m_i + \sum_{a \in [r']} m'_a .$$
From the bounds on $m_i$ and $m'_a$, we get
$$m \leq \sum_{i \in [r]} f(c-1, n_i) + \sum_{a \in [r']} f(c-1, n'_a) .$$
Recall that $n_i \geq \sqrt{n}$. 
We break the second sum, in the above equation,
into two parts.
Let $R_1 = \{a \in [r'] \mid n'_a \geq \sqrt{n} \}$ 
and $R_2 = \{ a \in [r'] \mid n'_a < \sqrt{n} \}$.
%Let $R_3 := R_1 \cup [r]$. 
%Let us merge first part of the second sum with the first sum and write, 
\begin{equation}
\label{eq:threeParts}
m \leq \sum_{i \in [r]} f(c-1, n_i) +  \sum_{a \in R_1} f(c-1, n'_a) + \sum_{a \in R_2} f(c-1, n'_a).
\end{equation}

%where $n_i \geq \sqrt{n}, \; \forall i \in R_3$. 
Let us first focus on the third sum. 
Note that $\abs{R_2} \leq r' < \sqrt{n}$.
For $a \in R_2$, $n'_a < \sqrt{n}$ and hence 
$f(c-1, n'_a) < f(c-1, \sqrt{n}) = 
2^{c-2} \cdot n^{1/2 - (1/2^{c-1})}$.
So, 
\begin{equation}
\label{eq:sum3}
\sum_{a \in R_2} f(c-1, n'_a) < \sqrt{n} \cdot 2^{c-2} \cdot n^{1/2 - (1/2^{c-1})} 
= 2^{c-2} \cdot n^{1 - (1/2^{c-1})} .
\end{equation}

Now, we focus on first two sums in Equation~\ref{eq:threeParts}. As, $n_i \geq \sqrt{n}, \; \forall i \in [r]$ and 
$n'_a \geq \sqrt{n}, \; \forall a \in R_1$, we combine these two sums (with an abuse of notation)
and write the sum as follows,
$$ \sum_{i \in [r'']} f(c-1, n_i) ,$$
where $r'' = r + \abs{R_1}$, and $n_i \geq \sqrt{n}, \; \forall i \in [r''] $. 
As each $n_i \geq \sqrt{n}$, we know $r'' < \sqrt{n}$ (as $\sum n_i \leq n$). 

Observe that $f(c-1, z)$, as a function of $z$, is a concave function (its derivative is monotonically decreasing, when $z > 0$). From the properties of a concave function, we know,
$$ \frac{1}{r''}\sum_{i \in [r'']} f(c-1, n_i) \leq 
       f \left( c-1, \frac{1}{r''} \sum_{ i \in [r'']} n_i \right).
$$
Now, $\sum_{i \in [r'']} n_i \leq n$ and $f(c-1, z)$ is an increasing function (when $z >0$). Hence, 
$$ \frac{1}{r''}\sum_{i \in [r'']} f(c-1, n_i) \leq 
       f \left( c-1, \frac{1}{r''} n \right).
$$
Equivalently,

\begin{eqnarray*} 
\sum_{i \in [r'']} f(c-1, n_i) &\leq &
      r''  \cdot 2^{c-2} \cdot (n/r'')^{1 - (1/2^{c-2})} \\
	&=& 2^{c-2} \cdot n^{1- (1/2^{c-2})} \cdot (r'')^{1/2^{c-2}} \\
	& < & 2^{c-2} \cdot n^{1- (1/2^{c-2})} \cdot n^{1/2^{c-1}} \\
	& = & 2^{c-2} \cdot n^{1- (1/2^{c-1})} 
\end{eqnarray*}

Using this with Equation~\ref{eq:sum3} and substituting in Equation~\ref{eq:threeParts},
we get
$$ m < 2^{c-1} \cdot n^{1- (1/2^{c-1})} .$$
\end{proof}

Now, we combine these results with our hitting-sets for 
depth-$3$ circuits having $m$ base sets with $\d$-distance. 

\begin{theorem}
Let $C(\o{x})$ be a $n$-variate polynomial, which is a sum of $c$ 
set-multinear depth-$3$ circuits, each having top fan-in $k$. 
Then there is a $n^{O(2^{c-1} n^{1 - \epsilon} \log k) }$-time whitebox test for $C$,
where $\epsilon := 1/2^{c-1}$.
\end{theorem}
\begin{proof}
As mentioned earlier, the polynomial $C(\o{x})$ can be viewed as 
computed by a depth-$3$ multilinear circuit, such that 
its product gates induce at most $c$-many distinct partitions. 
From Lemma~\ref{lem:cPartitions}, we can partition the variable set into
 $m$ base sets, such that for each of these base sets, the partitions can be sequenced
to have distance $1$, where $m := 2^{c-1} n^{1 - \epsilon}$.
Hence, the polynomial $C$ has $m$ base sets with $1$-distance and top fan-in $ck$.
Moreover, from the proof of Lemma~\ref{lem:cPartitions}, it is clear that 
such base sets can be computed in $n^{O(c)}$-time. 
From Theorem~\ref{thm:baseSetsHS}, we know there is $n^{O(m \log (ck))}$-time
whitebox test for such a circuit. 
 Substituting
the value of $m$, we get the result. 
\end{proof}

\end{document}